\documentclass[final,1p,times,authoryear]{elsarticle}

\usepackage{booktabs}
\usepackage{hyperref}

\makeatletter
\def\ps@pprintTitle{%
 \let\@oddhead\@empty
 \let\@evenhead\@empty
 \def\@oddfoot{}%
 \let\@evenfoot\@oddfoot}
\makeatother

\usepackage[utf8]{inputenc}
\usepackage{geometry}
\usepackage{amsmath, amssymb, amsthm}
\usepackage{graphicx}
\usepackage{natbib}
\usepackage{enumitem}
\usepackage{lineno}
\usepackage{subcaption}

\hypersetup{
    colorlinks=true,
    linkcolor=blue,
    filecolor=magenta,
    urlcolor=cyan,
    citecolor=blue,
}

\graphicspath{ {./graphics/} }

\newtheorem{proposition}{Proposition}
\newtheorem{lemma}{Lemma}
\newtheorem{corollary}{Corollary}

\newtheorem{assumption}{Assumption}

\DeclareMathOperator*{\argmin}{arg\,min}
\DeclareMathOperator{\Var}{Var}

\newcommand{\R}{\mathbb{R}}
\newcommand{\E}{\mathbb{E}}
\newcommand{\D}{\mathcal{D}}
\newcommand{\K}{\mathcal{K}}
\newcommand{\Lcal}{\mathcal{L}}
\newcommand{\Hcal}{\mathcal{H}}

\newcommand{\Jcal}{\mathcal{J}}

\begin{document}
\begin{frontmatter}

\title{Generalized Local Polynomial Regression with Decomposed Context-Aware Kernels}
\author{Yaniv Shulman}
\address{yaniv@shulman.info}
\date{\today}

\begin{abstract}
Local Polynomial Regression (LPR) is a powerful tool for nonparametric smoothing, yet it traditionally suffers from a "Euclidean tautology": the variables used to define the local neighborhood are identical to those used in the polynomial fit. This restricts its ability to handle complex domains where the regression function varies across non-Euclidean structures, such as graphs, manifolds, or discrete categories, while remaining locally smooth in the primary feature space. We propose \emph{Generalized Context-Aware LPR} (GC-LPR), a framework that decouples the fitting coordinates ($Z$) from the weighting context ($C$). By adopting a modeling convention where the conditional mean depends jointly on $Z$ and $C$ ($Y = m_C(Z) + \varepsilon$), our estimator acts as a "projected smoother": it isolates a slice of the data on the manifold defined by $C$ via a compound product kernel, and performs polynomial fitting in the $Z$-coordinates within that slice. This enables practitioners to model responses that vary across graphs, networks, or categorical strata while retaining the interpretability and bias properties of LPR in a primary Euclidean feature space. Theoretical analysis clarifies the induced context-smoothed target of GC-LPR and shows that the method preserves the Euclidean bias-reduction properties of standard LPR while allowing arbitrary, non-Euclidean contexts to modulate the local estimation. We demonstrate the efficacy of this approach on geospatial and network-structured datasets.
\end{abstract}
\end{frontmatter}

\section{Introduction}
\label{s:Introduction}

Local Polynomial Regression (LPR) is a cornerstone of nonparametric smoothing, prized for its ability to adapt locally without committing to a rigid global functional form. Its theoretical virtues are well-established, especially for local linear and higher-order variants: the local linear estimator is design-adaptive and automatically corrects first-order boundary bias, while the broader local-polynomial family attains the usual minimax-optimal rates under standard smoothness conditions \citep{fan1996local, loader1999local}. However, classical LPR suffers from a restrictive "Euclidean tautology": the variables used to define the local neighborhood are identical to those used in the polynomial fit.

This constraint is increasingly inadequate in modern domains where the notion of proximity is multi-faceted. Observations may be ``near'' geographically, connected via a transport network, or similar in terms of categorical attributes, even if they are distant in the feature space of interest. Furthermore, relying on a monolithic kernel over all available features exacerbates the ``curse of dimensionality'' \citep{stone1982optimal}, degrading performance as the dimension of the predictor space grows. Generalized product kernels have been successfully applied to mixed-data regression \citep{li2004nonparametric}, and GWR variants have already explored non-Euclidean notions of locality by replacing Euclidean distance with road-network distance, travel time, metro-network distance, or network-derived weighting schemes \citep{lu2011non_euclidean_gwr,lu2014non_euclidean_gwr,gao2022network_distance_gwr,he2023network_weight_gwr}. 

Our contribution is therefore more specific. We introduce \emph{Generalized Context-Aware LPR} (GC-LPR), a decomposed local-polynomial framework in which the polynomial fit is performed in primary Euclidean coordinates $Z$, while one or more auxiliary contexts $C_1,\ldots,C_J$ enter through a compound product kernel $K_c$ (formal definition in Section~\ref{s:Generalized Context-Aware LPR}). The Euclidean factor governs local polynomial fitting in $Z$, and the context factors restrict the neighborhood using domain-appropriate similarity notions, ranging from discrete labels to nodes on a graph. This yields an induced context-smoothed target $m_W$ and a bias-variance interpretation that separates polynomial approximation in $Z$ from smoothing across $C$. Accordingly, the asymptotic theory in Section~\ref{sec:theory} is stated for $m_W$ under fixed context bandwidths rather than for the point target $m_{c^\star}(z)$.

Conceptually, we do not treat context as affecting only the variance. Instead, we model the regression function itself as context-dependent, writing $m_C(Z)$. The context kernels then isolate a relevant slice of the data defined by $C$, within which the local polynomial fit is carried out in $Z$. This allows the relationship between $Z$ and $Y$ to change across the context space, whether smoothly or through structural discontinuities such as network or geographic boundaries. In this sense, GC-LPR is distinct from global graph-regularization methods such as Graph Signal Processing (GSP) \citep{shuman2013emerging}.

We further demonstrate the composability of this framework by integrating a response-based robustness factor $K_r$, following recent work by \citet{shulman2025robust}, to handle outliers without distorting the definition of locality. The remainder of this paper reviews the relevant literature, formalizes the estimator, analyzes its asymptotic properties, and validates its performance on geospatial and network-structured datasets.

\section{Background and Notation}
\label{s:Background}
This section establishes the formal notation and background for the methods discussed. We first review the canonical framework for Local Polynomial Regression (LPR) to define its objective and properties. Finally, we summarize the RSKLPR framework, which introduces a response-based kernel for robustness.

\subsection{Local Polynomial Regression}
\label{s:background_lpr}

Let $\D_{n}=\{(Z_i,Y_i)\}_{i=1}^{n}$ be i.i.d.\ with $Z_i\in\R^d$ and $Y_i\in\R$. Fix $z\in\R^d$ and write $t_i:=Z_i-z$. Let $\mathcal{A}_p=\{\alpha\in\mathbb{N}_0^d:\ |\alpha|\le p\}$, $t^\alpha=\prod_{j=1}^{d} t_j^{\alpha_j}$, and
\[
r_p(t):=(t^\alpha)_{\alpha\in\mathcal{A}_p}\in\R^{M},\qquad M=\binom{p+d}{d}.
\]
The local polynomial is $g_z(u;\beta)=r_p(u-z)^\top\beta(z)$, and the estimator of $m(z)$ is $\hat m(z):=\hat\beta_0(z)$, the intercept.

Let $K:\R^d\to[0,\infty)$ be a kernel and $H$ a positive-definite bandwidth matrix. With $K_H(t):=|H|^{-1}K(H^{-1}t)$, where $|H|:=\det(H)$ and in the isotropic case $H=hI_d$ gives $K_H(t)=h^{-d}K(t/h)$. 
In empirical objectives we use the unnormalized weight $K(H^{-1}t)$, since $|H|$ is constant in $i$, so it does not affect the minimizer. The empirical LPR objective is:
\begin{equation}
\label{eq:lpr_loss}
\Lcal_{lpr}(z;\D_n,H)=\sum_{i=1}^{n}\big(Y_i-r_p(t_i)^\top\beta(z)\big)^2 \, K(H^{-1}t_i).
\end{equation}
Let $R$ be the $n\times M$ design with $i$th row $r_p(t_i)^\top$, $W_z:=\mathrm{diag}\{K(H^{-1}t_1),\ldots,K(H^{-1}t_n)\}$, and $Y:=(Y_1,\ldots,Y_n)^\top$. When $R^\top W_z R$ is nonsingular,
\[
\hat\beta(z)=(R^\top W_z R)^{-1}R^\top W_z Y,\qquad \hat m(z)=e_1^\top\hat\beta(z).
\]
The population analogue (with the normalized kernel) is
\begin{equation}
\label{eq:lpr_pop}
\Jcal_{std}(z;\beta)=\iint\big(v-g_z(u;\beta)\big)^2 K_H(u-z)\, f_{ZY}(u,v)\, dv\,du.
\end{equation}
Under standard conditions (continuity of $f_Z$ with $f_Z(z)>0$; $m\in C^{p+1}$ near $z$; symmetric integrable $K$ with unit mass; $H\to0$ and $n|H|\to\infty$ in the isotropic case), the local linear estimator ($p=1$) is design-adaptive and, for interior $z$,
\[
\mathrm{Bias}\{\hat m(z)\}=O(\|H\|^{2}),\qquad \mathrm{Var}\{\hat m(z)\}=O\big((n|H|)^{-1}\big).
\]
See \citep{fan1993local,fan1996local,loader1999local} for details. Product and radial kernels both fit within this framework.

\subsection{Robust Local Polynomial Regression with Similarity Kernels}
\label{s:background_rsklpr}
Robust Local Polynomial Regression with Similarity Kernels (RSKLPR) \citep{shulman2025robust} augments LPR with a similarity kernel on the full data space to achieve robustness, in the spirit of classical robust local regression and local polynomial M-estimation \citep{cleveland1979robust,jiang2001robust}. Writing its empirical loss at $(z,y)$ as
\begin{equation}
\label{eq:rsk_loss}
\Lcal_{rsk}(z,y;\D_n,\Hcal)
=\sum_{i=1}^{n}\Big(Y_i-r_p(t_i)^\top\beta(z,y)\Big)^2\,
\K_D\big((z,y),(Z_i,Y_i);\Hcal\big),
\end{equation}

with the decomposition
\begin{equation}
\K_D\big((z,y),(u,v);\Hcal_1,\Hcal_2\big)
=K_1(z-u;\Hcal_1)\,K_2\big((z,y),(u,v);\Hcal_2\big),
\end{equation}
a separable choice $K_2((z,y),(u,v))=\hat f_{Y\mid Z}(y\mid z)\hat f_{Y\mid Z}(v\mid u)$ yields,
after dropping the multiplicative constant $\hat f_{Y \mid Z}(y \mid z)$ (independent of $i$),
a simplified empirical loss independent of $y$ at the query:
\begin{equation}
\label{eq:rsk_simple_loss}
\tilde{\Lcal}(z)
=\sum_{i=1}^{n}\Big(Y_i-r_p(t_i)^\top\beta(z)\Big)^2\,K_1(t_i;\Hcal_1)\,\hat f_{Y \mid Z}(Y_i \mid Z_i;\Hcal_2).
\end{equation}
In the population, the weight appears squared because it is multiplied by the conditional density itself as a weight inside the expectation:
\begin{equation}
\label{eq:rsk_pop}
\Jcal_{rsk}(z;\beta)
=\iint \big(v-g_z(u;\beta)\big)^2\,K_1(z-u;\Hcal_1)\,[f_{Y \mid Z}(v \mid u)]^2\,f_Z(u)\,dv\,du.
\end{equation}
The intercept targets $\mu'(z)=\frac{\int v[f_{Y \mid Z}(v \mid z)]^2\,dv}{\int [f_{Y \mid Z}(v \mid z)]^2\,dv}$, which reduces to $m(z)$ when $f_{Y \mid Z}(\cdot \mid z)$ is symmetric about its mean; see \citep{shulman2025robust}. \newline

\section{Generalized Context-Aware LPR}
\label{s:Generalized Context-Aware LPR}

\subsection{Modeling Convention: The Context-Dependent Mean}
\label{s:gc_convention}

Let $X \in \mathcal{X}$ denote a predictor observation where we emphasize that $\mathcal{X}$ is a general predictor space: an observation $X$ may include Euclidean features, categorical labels, spatial coordinates, or identifiers of nodes in a graph or network. We assume there exists a Euclidean coordinate chart $\psi_0:\mathcal{X}\to\R^d$ and additional feature maps yielding the decomposition $X \mapsto (Z,C)$, where $Z=\psi_0(X)\in \R^d$ are the fitting coordinates and $C$ collects additional \emph{contexts} (categorical, geospatial, network-valued, etc.). We do not treat $C$ merely as variance modifiers but rather adopt the conditional-mean formulation:
\[
m_c(z) \;\equiv\; \mathbb{E}\!\left[Y \mid Z = z,\, C = c\right].
\]
Accordingly, we write the observation model as
\[
Y = m_C(Z) + \varepsilon,
\qquad
\mathbb{E}[\varepsilon \mid Z,C] = 0, \quad \mathrm{Var}(\varepsilon \mid Z,C) = \sigma^2(Z,C).
\]
The underlying point target of interest at \( x^\star = (z, c^\star) \) is the context-specific regression function:
\[
m_{c^\star}(z).
\]

Here, the regression function $m_{c^\star}(z)$ is assumed to be smooth (locally polynomial) in $Z$, but its dependence on $C$ may be complex, discontinuous, or defined on a graph. The role of the context kernels is to isolate a local "slice" or neighborhood in the $C$-space, within which the polynomial approximation in $Z$ is valid. 

\subsection{Estimator}
We decouple the Euclidean fitting chart from the weighting context by defining feature maps over the general predictor space $\mathcal{X}$. We define the \emph{compound predictor kernel} as a product of individual kernels on mapped features:
\begin{equation}
\label{eq:kc_def}
K_c(x,u;\Hcal_c)
= K_{0,H}\!\big(\psi_0(x)-\psi_0(u)\big)\,
\prod_{j=1}^{J} K_j\!\big(\psi_j(x),\psi_j(u);\eta_j\big),
\qquad
\Hcal_c=(H,\eta_1,\ldots,\eta_J).
\end{equation}

Here, $\psi_0: \mathcal{X} \to \R^d$ denotes the Euclidean coordinate chart used for local polynomial fitting, and we write $Z=\psi_0(X)$ for the corresponding chart coordinates, equipped with the standard kernel $K_{0,H}(\cdot)=|H|^{-1}K_0(H^{-1}\cdot)$. For $j\ge1$, the maps $\psi_j: \mathcal{X} \to \mathcal{C}_j$ extract auxiliary context features $C_j \in \mathcal{C}_j$. The context kernels $K_j(\cdot;\eta_j)$ are symmetric on their respective spaces $\mathcal{C}_j$ and operate on the feature maps $\psi_j(X)$; they may incorporate non-Euclidean structure (e.g., Haversine distance for coordinates, graph shortest-path or diffusion distances \citep{shuman2013emerging}, or Aitchison--Aitken kernels for discrete labels \citep{aitchison1976multivariate}). This recovers classical product-kernel mixed-data smoothing while generalizing to non-Euclidean topologies.

With these components, the GC-LPR estimator $\hat\beta(x)$ minimizes the context-weighted least squares objective:
\begin{equation}
\label{eq:gc_lpr_loss}
\Lcal_{\text{gc-lpr}}(x)
=\sum_{i=1}^{n}\Big(Y_i-r_p(\psi_0(X_i)-\psi_0(x))^\top\beta(x)\Big)^2\,K_c(x,X_i;\Hcal_c).
\end{equation}
Equivalently, writing \(z=\psi_0(x)\) and \(Z_i=\psi_0(X_i)\), define the empirical local Gram matrix
\[
G_n(x)\coloneqq \sum_{i=1}^{n} K_c(x,X_i;\Hcal_c)\, r_p(Z_i-z)\,r_p(Z_i-z)^\top.
\]
When \(G_n(x)\) is nonsingular, the minimizer is unique and satisfies
\[
\hat\beta(x)=G_n(x)^{-1}\sum_{i=1}^{n} K_c(x,X_i;\Hcal_c)\, r_p(Z_i-z)\,Y_i.
\]
In computation one may replace \(G_n(x)^{-1}\) by a Moore--Penrose inverse if needed; for the pointwise asymptotic results below, we assume \(G_n(x^\star)\) is nonsingular with probability tending to one.
Crucially, the polynomial basis $r_p(\cdot)$ operates only on $Z_i=\psi_0(X_i)$, whereas the weight factorizes as
\[
K_c(x,X_i;\Hcal_c)=K_{0,H}(Z_i-z)\,W(x,X_i),
\qquad
W(x,u)\coloneqq \prod_{j=1}^{J} K_j\!\big(\psi_j(x),\psi_j(u);\eta_j\big),
\]
where $z=\psi_0(x)$ and $W$ involves only the maps $\{\psi_j\}_{j \ge 1}$. Note that while the polynomial basis $r_p$ varies only with $Z$, the estimated coefficients $\beta(x)$ vary with the full context $C$, allowing the shape of the regression function to change across contexts.

Specifically, given finite context bandwidths, the estimator naturally targets the \emph{context-smoothed mean}:
\[
m_W(z;x^\star) = \frac{\mathbb{E}[Y \, W(x^\star, X) \mid Z=z]}{\mathbb{E}[W(x^\star, X) \mid Z=z]},
\]
where $x^\star=(z,c^\star)$ with $z=\psi_0(x^\star)$. While analytically distinct from the pointwise target $m_{c^\star}(z)$, in domains such as graph signal processing this smoothing is often a desired feature rather than a bias term. It allows the model to ``borrow strength'' from topologically proximate observations, effectively denoising the structural dependency on $C$ while fitting the geometry of $Z$.

Hence, at $x^\star$, GC\text{-}LPR yields a local polynomial approximation to $m_W(z;x^\star)$; this coincides with $m_{c^\star}(z)$ under exact context selection or vanishing context bandwidth.

\paragraph{Example (geospatial and categorical contexts)}
Consider a raw observation $X$ consisting of tabular features $Z_{\text{feat}}$, geospatial coordinates $Z_{\text{geo}}$, a region label $C_{\text{cat}}$, and a network node $C_{\text{net}}$. We define the feature maps as:
\[
\psi_0(X)=Z_{\text{feat}} \oplus Z_{\text{geo}},\quad
\psi_1(X)=Z_{\text{geo}},\quad
\psi_2(X)=C_{\text{cat}},\quad
\psi_3(X)=C_{\text{net}},
\]
where $\oplus$ denotes concatenation. We then define the compound kernel $K_c$ explicitly in terms of these inputs:
\[
K_c(x,u;\Hcal_c) =
K_{0,H}\big(\psi_0(x)-\psi_0(u)\big) \cdot
K_{\text{geo}}\big(\psi_1(x),\psi_1(u);\eta_1\big)\cdot
K_{\text{cat}}\big(\psi_2(x),\psi_2(u);\eta_2\big)\cdot
K_{\text{net}}\big(\psi_3(x),\psi_3(u);\eta_3\big).
\]
This formulation allows the model to handle ``manifold folding'' (e.g., cul-de-sacs) where points are close in the Euclidean coordinates induced by $\psi_0$ but distant in the network context $\psi_3$. Specifically, $K_{\text{net}}$ will assign low weight to physically close but network-distant points, correctly estimating the mean for the specific network context. Note that $Z_{\text{geo}}$ appears in both the fitting chart $\psi_0$ and the context map $\psi_1$; GC-LPR allows variables to be shared between the regressor $r_p$ and the weighting kernel $K_c$. When overlaps are heavy, one may temper factors via exponents $K_j^{\alpha_j}$ with $\alpha_j\in(0,1]$, or tune the $(H,\eta_j)$ by cross-validation.

\subsection{A Robust Extension (GRC-LPR)}

Optional robustness against $Y$-outliers is incorporated by a separable response-based factor $K_r$, following \citet{shulman2025robust}. To control the curse of dimensionality in the density estimation, we define a robustness feature map $\psi_r: \mathcal{X} \to \R^{d_r}$ (typically $\psi_r(X)=\psi_0(X)=Z$) to select the variables relevant for detecting outliers. Define
\begin{equation}
\label{eq:grc_decomp}
\K_{\D}^*\big((x,y),(u,v);\Hcal\big)\;:=\;K_c(x,u;\Hcal_c)\;K_r\big((\psi_r(x),y),(\psi_r(u),v);\Hcal_r\big),
\end{equation}
and consider the empirical loss at $x$ obtained when $K_r$ is chosen as a separable conditional-density weight in the $\psi_r$-chart, i.e.,
$K_r\big((\psi_r(x),y),(\psi_r(u),v)\big)=\hat f_{Y\mid \psi_r(X)}(y\mid \psi_r(x);\Hcal_r)\,\hat f_{Y\mid \psi_r(X)}(v\mid \psi_r(u);\Hcal_r)$.
Dropping the multiplicative constant $\hat f_{Y\mid \psi_r(X)}(y\mid \psi_r(x);\Hcal_r)$ (independent of $i$) yields:
\begin{equation}
\label{eq:grc_simple_loss}
\tilde{\Lcal}_{\text{grc-lpr}}(x)
=\sum_{i=1}^{n}\Big(Y_i-r_p(\psi_0(X_i)-\psi_0(x))^\top\beta(x)\Big)^2\,K_c(x,X_i;\Hcal_c)\,\hat f_{Y\mid \psi_r(X)}(Y_i \mid \psi_r(X_i);\Hcal_r).
\end{equation}
This formulation yields a single-step reweighted LPR where points with low conditional density, relative to the chosen subspace $\psi_r(X)$, receive smaller weights. While one may choose $\psi_r(X)=X$, restricting $\psi_r$ to the fitting coordinates $Z$ or a low-dimensional subset may be preferred to ensure stable density estimates.

\section{Related Work}
\label{s:Related Work}

The proposed GC-LPR framework sits at the intersection of nonparametric smoothing, mixed-data regression, and signal processing on irregular domains. This section contextualizes our contribution within these established fields.

\subsection{Generalized Product Kernels and Mixed Data}
The foundational challenge in nonparametric regression is the ``curse of dimensionality.'' \citet{li2004nonparametric} revolutionized the treatment of mixed data by introducing generalized product kernels. Their estimator, $\hat{E}(Y|X^c, X^d)$, utilizes standard kernels for continuous variables ($X^c$) and variation-based kernels (e.g., Aitchison-Aitken) for discrete variables ($X^d$). A critical property of this framework is ``automatic relevance determination'': cross-validation can effectively shrink the bandwidths of irrelevant predictors to infinity, smoothing them out of the model \citep{hall2007nonparametric}.

However, the Racine-Li framework assumes that all variables entering the kernel also enter the regression function. GC-LPR relaxes this by decoupling the sets. We allow context variables $C$ to enter the kernel $K_c$ without forcing them into the polynomial basis $r_p(Z)$. This is broadly in line with work on associated kernels for mixed or support-constrained regressors \citep{some2016effects}, which emphasizes matching the kernel support to the support of each regressor. In GC-LPR, however, the context kernels serve as auxiliary locality weights rather than as additional regressors entering the polynomial mean directly.

\subsection{Varying Coefficient Models and GWR}
Varying Coefficient Models (VCM) extend the linear model by allowing coefficients to change smoothly as a function of an index variable, typically written as $Y = X^\top \beta(U) + \epsilon$ \citep{hastie1993varying, fan1999statistical}. \citet{he2009double} further refined VCM estimation using double-smoothing bias reduction techniques.

Geographically Weighted Regression (GWR) \citep{brunsdon1996geographically, fotheringham2003geographically} is perhaps the most famous instance of a spatial VCM. In classical GWR and related extensions such as Multiscale GWR (MGWR), one local regression is indexed by location and locality is defined through a spatial weighting scheme \citep{fotheringham2017multiscale}. Prior GWR work has already broadened this notion of locality beyond Euclidean distance, including road-network distance, travel time, metro-network distance, and network-derived weight matrices \citep{lu2011non_euclidean_gwr,lu2014non_euclidean_gwr,gao2022network_distance_gwr,he2023network_weight_gwr}. Those contributions primarily modify the distance metric or network weight matrix inside otherwise standard GWR machinery. A more theoretical GWR literature also studies estimation and inference, location-specific bandwidths, tests for locational heterogeneity and model specification, and local-linear bias reduction \citep{paez2002gwr_part1,paez2002gwr_part2,wang2008local_linear_gwr}. 

GC-LPR is related to VCM and GWR, but it is not a direct re-labeling of the context $C$ as the VCM index. Standard VCMs and GWR typically use a single indexing space to define the local fit. GC-LPR instead separates fitting coordinates $Z$ from auxiliary contexts $C_1,\ldots,C_J$, and allows several context kernels to enter multiplicatively while the polynomial basis remains in $Z$. In that sense, the framework is closer to a decomposed, context-indexed local polynomial model than to a standard spatial VCM. Our novelty is therefore not the use of non-Euclidean locality itself, but the decomposed multi-context local-polynomial formulation and the target-level theory built around $m_W$.

\subsection{Graph Signal Processing and Manifold Regularization}
When data resides on irregular domains, Graph Signal Processing (GSP) and manifold regularization address smoothness via operators such as the graph Laplacian \citep{shuman2013emerging}. Diffusion kernels (or heat kernels) $K_t = e^{-t\mathcal{L}}$ generalize Gaussian smoothing to manifolds and graphs, capturing the intrinsic geometry of the data domain \citep{kondor2002diffusion, belkin2003laplacian, coifman2006diffusion}. These methods are conceptually distinct from local regression: they are typically global graph-smoothing or spectral procedures rather than weighted local polynomial fits.

Work by \citet{bickel2007local} extended local polynomial regression to unknown manifolds, proving that LPR can adapt to the intrinsic dimension of the data. GC-LPR is complementary to these spectral methods. By incorporating graph-based kernels (e.g., diffusion or shortest-path) as multiplicative factors in $K_c$, we combine local-polynomial fitting in $Z$ with graph-informed context weighting. This allows us to estimate the context-dependent mean $m_C(Z)$ even when $C$ represents a complex manifold structure that is not isometric to Euclidean space.

In summary, GC-LPR can be viewed as a decomposed local-polynomial framework related to mixed-data smoothing, GWR, and graph-based learning, but distinct in its separation of fitting coordinates from multiple auxiliary contexts and in its target-level interpretation via $m_W$.

\section{Theoretical Properties}
\label{sec:theory}

We study the estimator in \eqref{eq:gc_lpr_loss} at a fixed query point \(x^\star\in\mathcal{X}\).
Write \(z=\psi_0(x^\star)\in\mathbb{R}^d\), and let \(c^\star\) denote the associated context value of \(x^\star\) under the decomposition \(X\mapsto (Z,C)\). Assume that \(z\) is an interior point of the support of \(Z\). When discussing population quantities,
we use the regular conditional distribution of \((X,Y)\) given \(Z=\psi_0(X)\), and write
\(\mathbb{E}[\cdot\mid Z=u]\) for conditional expectations.
Throughout this section we assume i.i.d.\ sampling of \(\{(X_i,Y_i)\}_{i=1}^n\). Consequently, the asymptotic results below do not cover graph-dependent or graph-temporal dependence structures such as those used in Experiments~3 and~4; those experiments should therefore be read as empirical demonstrations outside the scope of the present theory.
Throughout this section, the context bandwidths \(\eta_1,\ldots,\eta_J\) are treated as fixed while the Euclidean bandwidth matrix \(H\) shrinks. Accordingly, the estimand of the asymptotic analysis is the induced context-smoothed target \(m_W(\cdot;x^\star)\); consistency for the point target \(m_{c^\star}(z)\) would require a separate theorem with \(\eta_j\to0\) and compatible joint rate conditions.

For clarity, recall the context weight and its conditional mean
\[
W(x,u)\coloneqq \prod_{j=1}^J K_j\!\big(\psi_j(x),\,\psi_j(u);\eta_j\big),
\qquad
\gamma_x(u)\coloneqq \mathbb{E}\!\left[\,W(x,X)\mid Z=u\,\right],
\]
so that the compound kernel factorizes into Euclidean and context components:
\[
K_c(x,u)\;=\;K_{0,H}\!\big(\psi_0(x)-\psi_0(u)\big)\,W(x,u).
\]
Whenever \(\gamma_x(u)>0\) and the displayed conditional expectations are finite, define the associated context-smoothed conditional mean
\[
m_W(u;x)\coloneqq
\frac{\mathbb{E}[Y\,W(x,X)\mid Z=u]}{\mathbb{E}[W(x,X)\mid Z=u]}
=\frac{\mathbb{E}[Y\,W(x,X)\mid Z=u]}{\gamma_x(u)}.
\]

\begin{lemma}[Context-smoothed conditional target]
\label{lem:mw_target}
Fix \(x\in\mathcal{X}\) and \(u\in\mathbb{R}^d\) such that \(0<\gamma_x(u)<\infty\) and
\[
\mathbb{E}\!\left[Y^2\,W(x,X)\mid Z=u\right]<\infty.
\]
Then
\[
m_W(u;x)
=
\argmin_{\alpha\in\mathbb{R}}
\mathbb{E}\!\left[(Y-\alpha)^2\,W(x,X)\mid Z=u\right].
\]
Equivalently, for every measurable function \(a:\mathbb{R}^d\to\mathbb{R}\),
\[
\mathbb{E}\!\left[(Y-a(Z))^2\,W(x,X)\mid Z=u\right]
=
\mathbb{E}\!\left[(Y-m_W(u;x))^2\,W(x,X)\mid Z=u\right]
+ \gamma_x(u)\,\big(a(u)-m_W(u;x)\big)^2.
\]
In particular, for each fixed \(u\), \(m_W(u;x)\) is the unique conditional weighted \(L^2\)-projection onto constants; equivalently, \(m_W(\cdot;x)\) is the pointwise conditional weighted \(L^2\)-projection target induced by the context weights.
\end{lemma}
Proof in~\ref{app:proof_mw_target}. Under the modeling convention \(m_c(z)=\mathbb{E}[Y\mid Z=z,\,C=c]\), our estimand at
\(x^\star\) is the context-specific mean \(m_{c^\star}(z)\). However, with finite context bandwidths,
the induced population target is the context-smoothed mean \(m_W(z;x^\star)\); the difference
\(m_W(z;x^\star)-m_{c^\star}(z)\) captures the effect of smoothing (or borrowing strength) across contexts.

\subsection{Assumptions}
\label{s:assumptions}

\begin{assumption}[Data, design, and local smoothness]
\label{ass:data}
Let \(\{(X_i,Y_i)\}_{i=1}^n\) be i.i.d.\ copies of \((X,Y)\). Let \(Z=\psi_0(X)\in\mathbb{R}^d\) admit a density \(f_Z\) that is continuous in a neighborhood of \(z\),
with \(f_Z(z)>0\). Assume the regression model
\[
Y=m_C(Z)+\varepsilon,\qquad \mathbb{E}[\varepsilon\mid Z,C]=0,\qquad \Var(\varepsilon\mid Z,C)=\sigma^2(Z,C).
\]
Fix \(x^\star\), and assume the context-smoothed mean \(m_W(\cdot;x^\star)\in C^{p+1}\) in a neighborhood of \(z\).
In particular, if the context kernels induce no leakage away from \(c^\star\) so that \(m_W(\cdot;x^\star)=m_{c^\star}(\cdot)\)
locally, this reduces to \(m_{c^\star}(\cdot)\in C^{p+1}\) near \(z\).
\end{assumption}

\begin{assumption}[Euclidean kernel and bandwidth]
\label{ass:k0}
\(K_0:\mathbb{R}^d\to[0,\infty)\) is even, integrable with finite moments up to order \(p+2\), and
\(\int K_0(t)\,dt=1\). The bandwidth \(H\) is positive definite with \(\|H\|\to0\) and \(n|H|\to\infty\).
Moreover, the kernel moment matrices
\[
M_0\coloneqq \int r_p(t)\,r_p(t)^\top K_0(t)\,dt,
\qquad
\Omega_0\coloneqq \int r_p(t)\,r_p(t)^\top K_0(t)^2\,dt
\]
have finite entries, and \(M_0\) is nonsingular.
\end{assumption}

\begin{assumption}[Context factors and effective-weight regularity]
\label{ass:context}
For \(j\ge 1\), the context factor \(K_j(\cdot,\cdot;\eta_j)\) is symmetric and nonnegative, and is
uniformly bounded above: \(0\le K_j(\cdot,\cdot;\eta_j)\le \bar K_j<\infty\). (Equivalently, one may replace each
factor by \(K_j/\bar K_j\in[0,1]\), which only rescales the weighted least-squares objective by a constant and does
not change its minimizer.)

Fix \(x^\star\in\mathcal{X}\). Assume the conditional mean weight
\[
\gamma_{x^\star}(u)=\mathbb{E}[W(x^\star,X)\mid Z=u]
\]
is locally Lipschitz at \(z\); i.e., there exist constants \(L<\infty\) and \(\delta>0\) such that
\[
|\gamma_{x^\star}(u)-\gamma_{x^\star}(z)| \le L\|u-z\| \quad \text{whenever } \|u-z\|<\delta.
\]
This regularity requirement is imposed on the effective conditional mean weight \(\gamma_{x^\star}\) itself, regardless of whether the underlying context variables are continuous, discrete, or mixed. Locally constant cases are included as a special case.
\end{assumption}

\begin{assumption}[Positivity]
\label{ass:positivity}
Let \(\phi_1(x^\star)\coloneqq \mathbb{E}[W(x^\star,X)\mid Z=z]=\gamma_{x^\star}(z)\). Assume
\(0<\underline{\phi}\le \phi_1(x^\star)\).
\end{assumption}
Combined with Assumptions~\ref{ass:data} and \ref{ass:context}, this pointwise positivity implies that, after possibly shrinking the neighborhood of \(z\), both \(\gamma_{x^\star}(u)\) and the effective design density \(q_{x^\star}(u)\coloneqq f_Z(u)\gamma_{x^\star}(u)\) are bounded away from zero near \(z\).

\begin{assumption}[Effective residual second-moment regularity]
\label{ass:variance_moment}
Fix \(x^\star\in\mathcal{X}\), and define
\[
\nu_{x^\star}(u)\coloneqq
\mathbb{E}\!\left[\,W(x^\star,X)^2\,\big(Y-m_W(u;x^\star)\big)^2\mid Z=u\,\right].
\]
Assume \(\nu_{x^\star}\) is locally bounded in a neighborhood of \(z\) and continuous at \(z\). Equivalently, on the event \(\{Z=u\}\), one may write
\[
\nu_{x^\star}(u)=
\mathbb{E}\!\left[\,W(x^\star,X)^2\,\big(Y-m_W(Z;x^\star)\big)^2\mid Z=u\,\right].
\]
\end{assumption}

\subsection{Bias and Variance}

\begin{lemma}[Effective population LPR problem]
\label{lem:weighted_target}
Fix a query point \(x\in\mathcal{X}\) and write \(z=\psi_0(x)\). The population objective
\[
\Jcal_{\mathrm{gc-lpr}}(x;\beta)
:= \mathbb{E}\Big[(Y - r_p(Z-\psi_0(x))^\top\beta(x))^2\,K_{0,H}(Z-\psi_0(x))\,W(x,X)\Big]
\]
is, up to an additive constant independent of \(\beta\), equal to
\[
\tilde{\Jcal}_{\mathrm{gc-lpr}}(x;\beta)
=
\int \Big(m_W(u;x)-r_p(u-z)^\top\beta(x)\Big)^2\,K_{0,H}(u-z)\,\gamma_x(u)\,f_Z(u)\,du.
\]
In particular, GC-LPR is (population-wise) standard local polynomial regression for the target \(m_W(\cdot;x)\)
under the effective design density \(u\mapsto f_Z(u)\,\gamma_x(u)\).
\end{lemma}
\begin{proof}
Apply Lemma~\ref{lem:mw_target} with \(a(Z)=r_p(Z-\psi_0(x))^\top\beta(x)\), multiply by \(K_{0,H}(Z-\psi_0(x))\), and take expectations. The additive term
\[
\mathbb{E}\Big[(Y-m_W(Z;x))^2\,K_{0,H}(Z-\psi_0(x))\,W(x,X)\Big]
\]
does not depend on \(\beta\), and the remaining term reduces by iterated expectation to
\[
\int \Big(m_W(u;x)-r_p(u-z)^\top\beta(x)\Big)^2\,K_{0,H}(u-z)\,\gamma_x(u)\,f_Z(u)\,du.
\]
\end{proof}

\begin{proposition}[Bias decomposition and rate]
\label{prop:bias}
Under Assumptions \ref{ass:data}--\ref{ass:positivity}, fix \(x^\star\in\mathcal{X}\). The GC-LPR intercept estimator
\(\hat m(x^\star)\coloneqq e_1^\top \hat\beta(x^\star)\), where \(e_1=(1,0,\dots,0)\) selects the intercept term,
has induced population target \(m_W(z;x^\star)\). Moreover, its bias relative to the
context-specific mean \(m_{c^\star}(z)\) decomposes as
\[
\mathbb{E}\big[\hat m(x^\star)\big]-m_{c^\star}(z)
=
\underbrace{\Big(\mathbb{E}\big[\hat m(x^\star)\big]-m_W(z;x^\star)\Big)}_{\text{polynomial bias (in \(Z\))}}
\;+\;
\underbrace{\Big(m_W(z;x^\star)-m_{c^\star}(z)\Big)}_{\text{Smoothing Bias (from \(C\))}}.
\]
The first term is the standard local-polynomial approximation error in \(Z\), and satisfies
\[
\mathbb{E}\big[\hat m(x^\star)\big]-m_W(z;x^\star)
=O\!\big(\|H\|^{p+1}\big).
\]
Hence,
\[
\mathbb{E}\big[\hat m(x^\star)\big]-m_{c^\star}(z)
=
O\!\big(\|H\|^{p+1}\big)\;+\;\Big(m_W(z;x^\star)-m_{c^\star}(z)\Big).
\]
In particular, if the context kernels enforce exact selection of the context cell containing \(c^\star\)
(e.g., \(W(x^\star,X)=0\) outside that cell, so that \(m_W(z;x^\star)=m_{c^\star}(z)\)), then the smoothing bias
term vanishes and the \(O(\|H\|^{p+1})\) rate is preserved. In the continuous-context case, sharpening the
context kernels (smaller \(\eta\)) can reduce \(\big|m_W(z;x^\star)-m_{c^\star}(z)\big|\), typically at the
cost of increased variance through a smaller effective sample size.
The displayed \(O(\|H\|^{p+1})\) rate is a generic nonsharp umbrella upper bound inherited from the effective local polynomial problem; sharper interior rates depend on the usual parity refinements and are not pursued here.
\end{proposition}
Proof in~\ref{app:proof_bias}.

\begin{proposition}[Asymptotic variance]
\label{prop:variance}
Under Assumptions \ref{ass:data}--\ref{ass:variance_moment}, fix \(x^\star\in\mathcal{X}\) such that \(z=\psi_0(x^\star)\) is an interior point of the support of \(Z\). Then
\[
\mathrm{Var}\{\hat m(x^\star)\}
=
\frac{1}{n|H|}\,\frac{1}{f_Z(z)}\;
e_1^\top M_0^{-1}\,\Omega_0\,M_0^{-1} e_1\ \times\ \Phi(x^\star)
\;+\;o\!\big((n|H|)^{-1}\big),
\]
where \(M_0\) and \(\Omega_0\) are the kernel moment matrices from Assumption~\ref{ass:k0}, and
\[
\Phi(x^\star)=\frac{\mathbb{E}\big[\,W(x^\star,X)^2\,\big(Y-m_W(Z;x^\star)\big)^2\mid Z=z\,\big]}
{\mathbb{E}\big[\,W(x^\star,X)\mid Z=z\,\big]^2}
\]
is the effective residual factor. Under the model \(Y=m_C(Z)+\varepsilon\) with \(\mathbb{E}[\varepsilon\mid Z,C]=0\),
\[
\Phi(x^\star)=
\frac{\mathbb{E}\big[\,W(x^\star,X)^2\big\{\sigma^2(z,C)+\big(m_C(z)-m_W(z;x^\star)\big)^2\big\}\mid Z=z\,\big]}
{\mathbb{E}\big[\,W(x^\star,X)\mid Z=z\,\big]^2}.
\]
\end{proposition}

Proof in~\ref{app:proof_variance}. The effective residual factor \(\Phi(x^\star)\) quantifies the statistical cost of context isolation and context mixing:
\begin{itemize}
    \item \textbf{Recovery of standard LPR:} If \(W\equiv 1\) (no context weights) and errors are
    homoskedastic and the model reduces to \(Y=m(Z)+\varepsilon\), then \(\Phi(x^\star)=\sigma^2\), recovering the standard LPR variance expression.
    \item \textbf{Context-mixing variation:} Even when \(\sigma^2(z,C)=0\), \(\Phi(x^\star)\) can remain positive if multiple contexts contribute nontrivially to the target \(m_W(\cdot;x^\star)\) at a given \(Z=z\), through the term \(\big(m_C(z)-m_W(z;x^\star)\big)^2\).
    \item \textbf{Weight-induced inflation:} When context weights are sparse or highly variable (e.g., focusing on a small manifold slice), the squared-weight numerator relative to \(\mathbb{E}[W\mid Z=z]^2\) can inflate variance, especially when the effective residual variation is roughly homogeneous or only weakly related to the weights. This reflects the reduction in effective sample size caused by the context filter. GC-LPR therefore does not eliminate the curse of dimensionality; rather, it trades model bias against variance. While the method allows for complex dependencies, a low-variance estimate still requires a sufficiently dense sample in the joint $(Z,C)$ space.
\end{itemize}

\begin{corollary}[Consistency for the context-smoothed target]
\label{cor:mw_consistency}
Under Assumptions \ref{ass:data}--\ref{ass:variance_moment}, fix \(x^\star\) such that \(z=\psi_0(x^\star)\) is an interior point of the support of \(Z\). Then
\[
\mathbb{E}\!\left[\big(\hat m(x^\star)-m_W(z;x^\star)\big)^2\right]
=
O\!\big(\|H\|^{2(p+1)}\big)+O\!\big((n|H|)^{-1}\big).
\]
Consequently, if \(\|H\|\to0\) and \(n|H|\to\infty\), then
\[
\hat m(x^\star)\xrightarrow{L^2} m_W(z;x^\star),
\qquad
\hat m(x^\star)\xrightarrow{P} m_W(z;x^\star).
\]
\end{corollary}
\begin{proof}
Proposition~\ref{prop:bias} gives \(\mathbb{E}[\hat m(x^\star)]-m_W(z;x^\star)=O(\|H\|^{p+1})\), while Proposition~\ref{prop:variance} and Assumptions~\ref{ass:positivity}--\ref{ass:variance_moment} give \(\Var\{\hat m(x^\star)\}=O((n|H|)^{-1})\); combine these with \(\mathbb{E}[(\hat m-m_W)^2]=\Var(\hat m)+(\mathbb{E}[\hat m]-m_W)^2\).
\end{proof}

Taken together, the results above give a pointwise-in-\(x^\star\) i.i.d.\ asymptotic theory for the bandwidth-matrix formulation of GC-LPR around the induced target \(m_W(z;x^\star)\). Under \(\|H\|\to0\) and \(n|H|\to\infty\), Corollary~\ref{cor:mw_consistency} yields consistency for \(m_W(z;x^\star)\), but not, by itself, for the context-specific mean \(m_{c^\star}(z)\); the latter would require shrinking context bandwidths and regularity across contexts. The theory also does not cover the adaptive \(k\)-nearest-neighbor implementation used in the experiments below, nor the additional response-based robustness factor used in GRC-LPR. Extending the analysis to those settings is left for future work.

\section{Applications and Experiments}
\label{s:Experiments}
The full experiment code, figure-generation scripts, and manuscript source are available in the GC-LPR repository at \url{https://github.com/yaniv-shulman/gclpr}.
All empirical results use the public Python implementation in the \texttt{rsklpr} package \citep{shulman2026rsklpr}, using the version 2.0 series.
The empirical goal here is to compare GC/GRC-LPR against feature-only local polynomial smoothers and simple generic baselines, rather than to claim a benchmark against all task-specific spatial or graph regression model classes.
\subsection{Experiment 1: Geospatial Context on California Housing}

As a foundational validation of the proposed framework, we test the hypothesis that explicit context kernels improve performance even when contextual variables (coordinates) are already present in the primary feature set. This is intended as a within-family comparison against feature-only local smoothers rather than a full ablation of every geographic representation choice. We evaluate GC-LPR against standard LPR, its robust variant RSKLPR \citep{shulman2025robust}, and a K-Nearest Neighbors (KNN) baseline.

\paragraph{\textbf{Dataset and Setup}}
We utilize the California Housing dataset \citep{pace1997sparse}, modeling the median house value (\texttt{MedHouseVal}) using a feature vector $Z$ containing \texttt{MedInc}, \texttt{AveRooms}, \texttt{Latitude}, and \texttt{Longitude}. Non-geospatial features were standardized. We use a repeated hold-out protocol with 5 repetitions of an \(80\%/20\%\) train/test split and report mean and standard deviation over the held-out test results.

For each repetition, hyperparameters are selected by 4-fold cross-validation on the training split only, and the selected model is then evaluated once on the held-out test split. The hyperparameter search space for local polynomial models included neighborhood size ($k \in [7,151]$), polynomial degree ($p \in \{0,1\}$), Laplacian and Tricube predictor kernels, and Minkowski or Mahalanobis distances. For GC-LPR and GRC-LPR, the search also optimized the length-scale of the Haversine context kernel. Additional evaluation details are collected in~\ref{app:exp_details}. The primary comparison is between:
\begin{itemize}
    \item \textbf{LPR / RSKLPR (Standard):} Utilizes a single features kernel over the full feature vector $Z$. 
    \item \textbf{GC-LPR / GRC-LPR (Geospatial):} Employs the compound kernel $K_c = K_{\text{std}}(Z) \cdot K_{\text{geo}}(\text{Lat, Lon})$, where $K_{\text{geo}}$ is a Haversine-based RBF kernel.
\end{itemize}

\paragraph{\textbf{Results and Statistical Analysis}}
The results are summarized in Table \ref{tab:results_california_5fold}. Both context-aware models outperform the feature-only local smoothers and KNN. GRC-LPR (the robust, context-aware variant) achieves the best mean performance across all three reported metrics, with RMSE $0.4993 \pm 0.0074$, MAE $0.3324 \pm 0.0034$, and $R^2 = 0.8134 \pm 0.0054$. GC-LPR is the second-best model, while the feature-only LPR family remains clearly ahead of KNN. The predictions scatter plot is shown in Figure \ref{figure:california_housling_exp_1_scatter_grid}, and the split-wise RMSE distributions are shown in Figure \ref{figure:california_exp1_rmse_boxplot}.

\begin{figure}[ht]
\centering
{\includegraphics[width=1.0\textwidth]{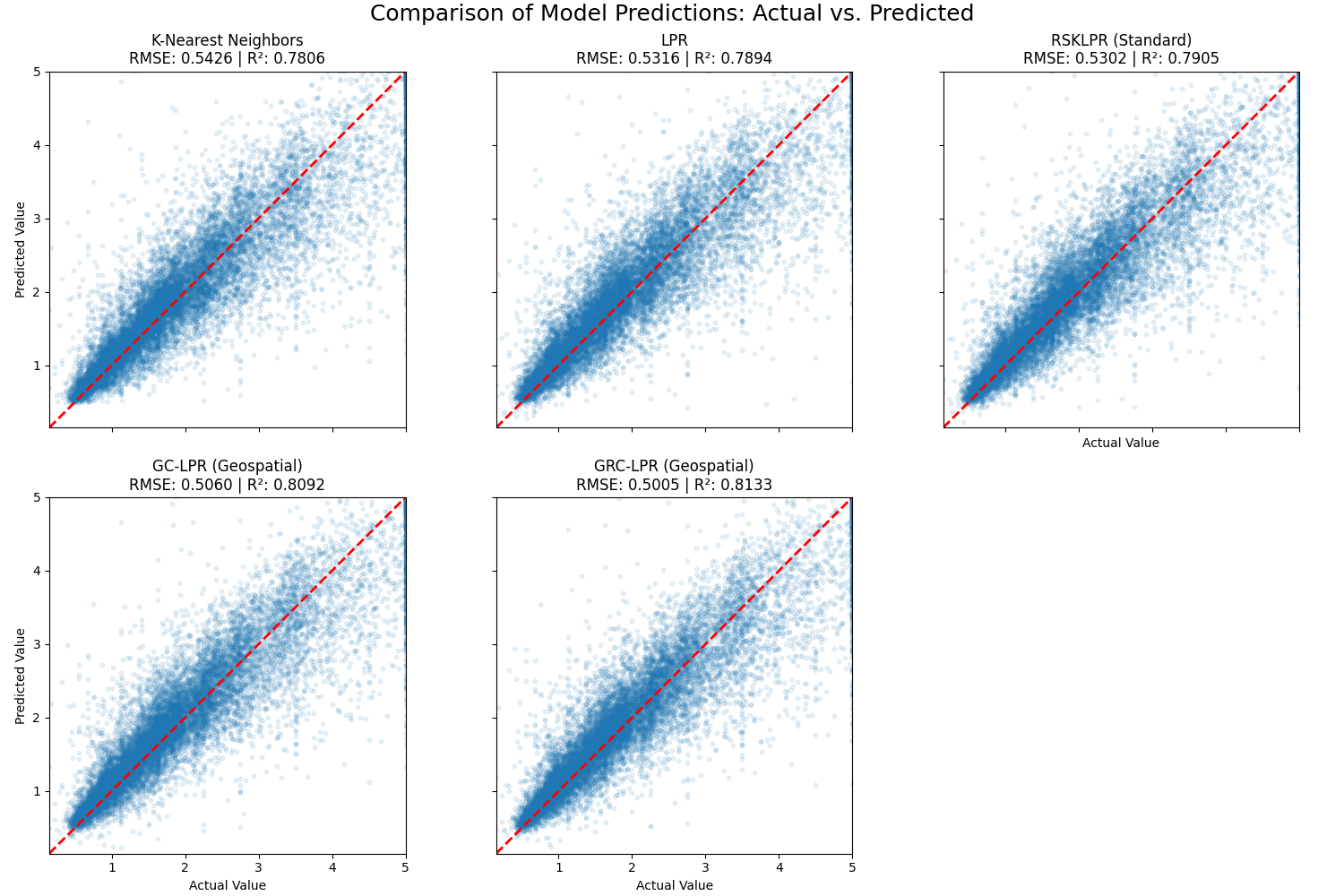}}
\caption{Comparison of model predictions for Experiment 1, aggregating the held-out predictions across the 5 repeated 80/20 train/test splits. The y-axis is the predicted value and x-axis is the actual value. The dashed line represents ideal predictions. The context-aware variants show tighter clustering around the ideal line than the feature-only smoothers.}
\label{figure:california_housling_exp_1_scatter_grid}
\end{figure}

\begin{table}[h]
\centering
\caption{Experiment 1 results on the California Housing dataset under 5 repeated 80/20 hold-out splits with 4-fold cross-validation on each training split. Values are mean \(\pm\) standard deviation over held-out test metrics.}
\label{tab:results_california_5fold}
\small
\begin{tabular}{@{}lrrr@{}}
\toprule
\textbf{Model} & \textbf{RMSE} & \textbf{MAE} & \textbf{$R^2$} \\ \midrule
\textbf{GRC-LPR (Geospatial)} & \textbf{0.4993 \(\pm\) 0.0074} & \textbf{0.3324 \(\pm\) 0.0034} & \textbf{0.8134 \(\pm\) 0.0054} \\
GC-LPR (Geospatial) & 0.5060 \(\pm\) 0.0098 & 0.3398 \(\pm\) 0.0051 & 0.8084 \(\pm\) 0.0067 \\
RSKLPR (Standard) & 0.5267 \(\pm\) 0.0100 & 0.3501 \(\pm\) 0.0055 & 0.7923 \(\pm\) 0.0090 \\
LPR (Standard) & 0.5280 \(\pm\) 0.0066 & 0.3536 \(\pm\) 0.0040 & 0.7914 \(\pm\) 0.0063 \\
KNN & 0.5423 \(\pm\) 0.0117 & 0.3562 \(\pm\) 0.0055 & 0.7798 \(\pm\) 0.0118 \\ \bottomrule
\end{tabular}
\end{table}

\begin{figure}[ht]
\centering
{\includegraphics[width=0.85\textwidth]{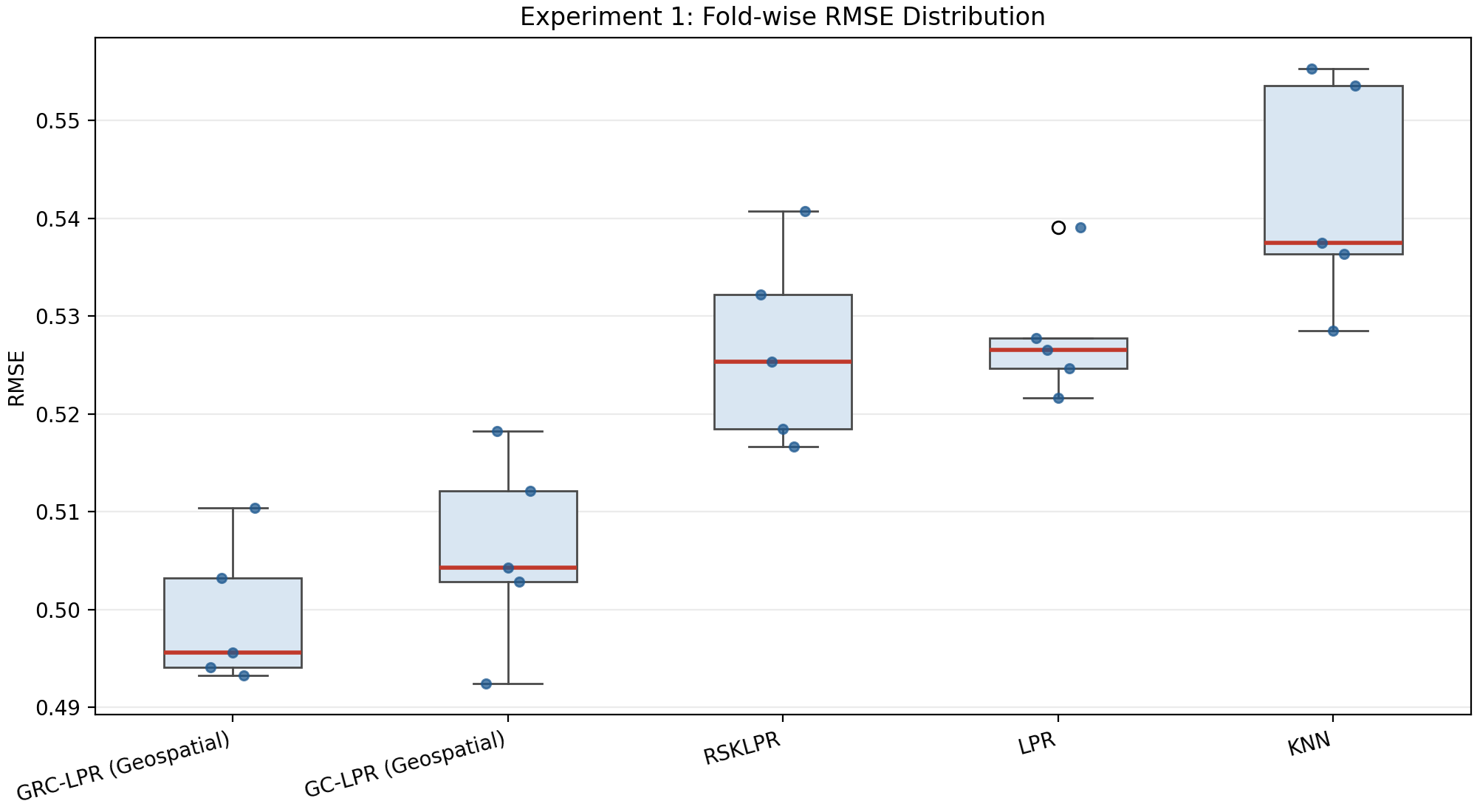}}
\caption{RMSE distributions for Experiment 1 under the repeated 80/20 hold-out protocol. The context-aware variants shift the full RMSE distribution downward, with GRC-LPR achieving the lowest center and the strongest average performance across held-out splits.}
\label{figure:california_exp1_rmse_boxplot}
\end{figure}

The explicit decoupling of geospatial context leads to a marked improvement over standard LPR (RMSE $0.5060$ vs.\ $0.5280$). Incorporating the response-based robustness factor ($K_r$) improves the context-aware model further, reducing RMSE from $0.5060$ to $0.4993$ and MAE from $0.3398$ to $0.3324$. In contrast, the robustness factor alone offers only a modest gain in the feature-only setting (RSKLPR vs.\ LPR).

\paragraph{\textbf{Spatial Error Distribution}}
To investigate the qualitative impact of context-awareness, we mapped the prediction errors across California (Figure \ref{figure:california_geospatial_errors}). The bottom row, corresponding to GC-LPR and GRC-LPR, is visibly more uniform than the feature-only local smoothers in the top row. The context-aware models better adapt to the distinct price dynamics of the San Francisco Bay Area and Greater Los Angeles, whereas the feature-only smoothers still over-smooth across these spatial boundaries.

\begin{figure}[ht]
\centering
{\includegraphics[width=1.0\textwidth]{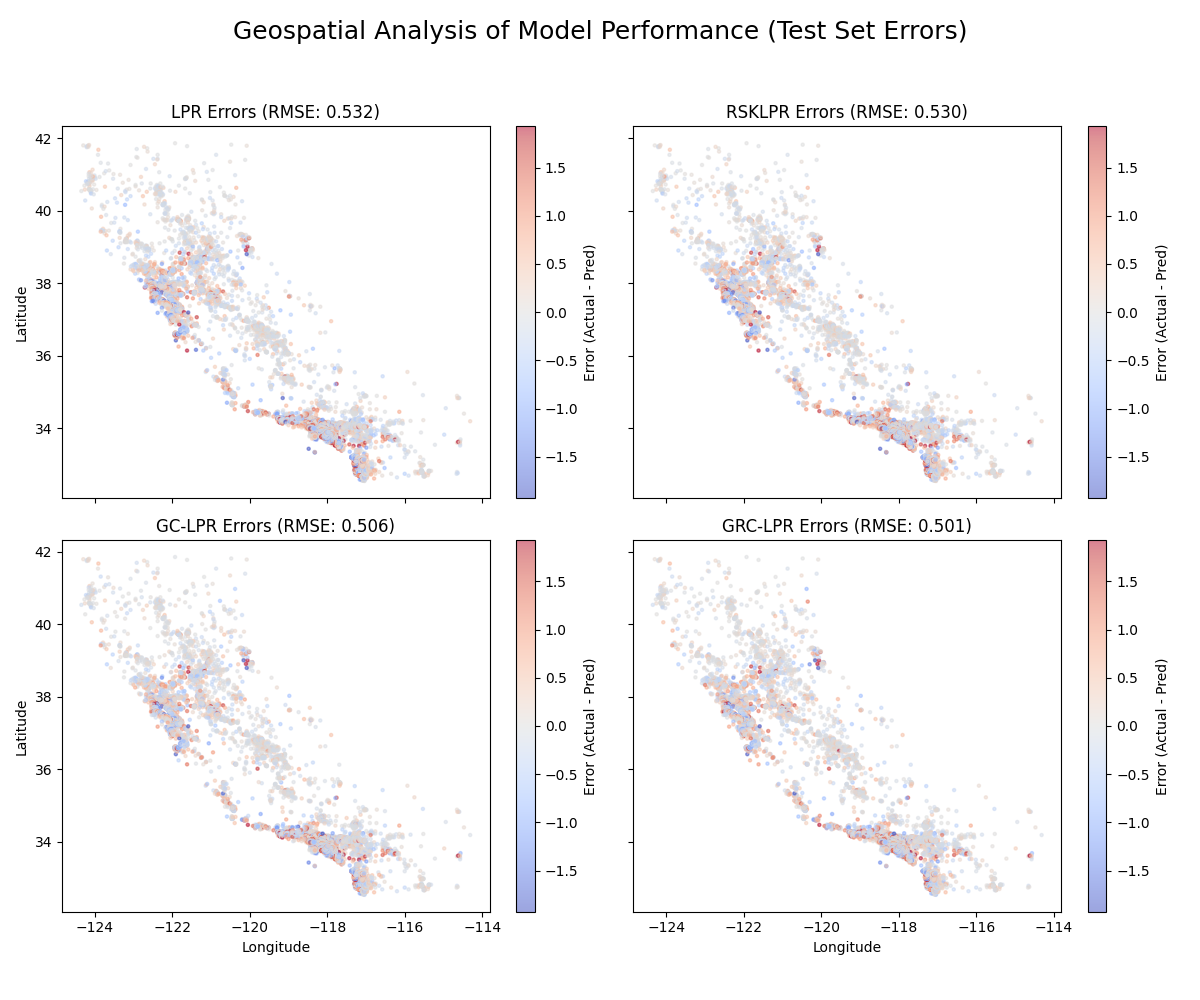}}
\caption{Geospatial distribution of prediction errors across California, aggregating the held-out predictions across the 5 repeated 80/20 train/test splits. The top row shows feature-only local smoothers (LPR and RSKLPR), while the bottom row shows the context-aware variants (GC-LPR and GRC-LPR). The context-aware models exhibit lower and more spatially uniform errors, especially in high-variance coastal urban centers.}
\label{figure:california_geospatial_errors}
\end{figure}

\paragraph{\textbf{Effective Neighborhood Analysis}}
Across the 5 repeated hold-out runs, the selected neighborhoods remained noticeably larger for the context-aware variants than for the feature-only local smoothers: \(k \in \{117,127,137\}\) for GC-LPR and \(k \in \{117,137,147\}\) for GRC-LPR, compared with \(k \in \{57,67\}\) for LPR and \(k \in \{27,37\}\) for RSKLPR. The local-polynomial models uniformly selected \(p=1\). This is consistent with the intended role of the context kernel: it allows the estimator to borrow strength from a wider Euclidean neighborhood while still down-weighting geographically irrelevant observations. The result is a more stable local fit without the same degree of cross-region smoothing bias.

\subsection{Experiment 2: Fusing Tabular, Geospatial, and Network Context in NYC}
The purpose of the second experiment is to demonstrate the use of a more advanced geospatial kernel that takes similar inputs as in Experiment 1 (i.e., \texttt{Latitude} and \texttt{Longitude}), but employs a non-Euclidean kernel on a graph. The goal is to show that standard distance metrics (Euclidean or Haversine) are often insufficient, and that a context-aware model can capture a more meaningful, network-based definition of proximity.

\paragraph{\textbf{Dataset and Setup}}
We use the archived NYC detailed-listings snapshot dated February 13, 2026 from Inside Airbnb \citep{insideairbnb2026nyc} to predict listing \texttt{price}. The features include standard tabular attributes (\texttt{host\_listings\_count}, \texttt{room\_type}, etc.) as well as \texttt{Latitude} and \texttt{Longitude}. The subway context graph is reconstructed from the MTA New York City Transit static GTFS subway feed \citep{mta_gtfs_subway_2026}. We use a repeated hold-out protocol: for each of 5 repetitions, the data are split into \(80\%\) training and \(20\%\) test sets, hyperparameters are selected by 4-fold cross-validation on the training split, and the selected model is then evaluated once on the held-out test split. The local models are fit on \(\log(1+\texttt{price})\), but tuning and reported metrics are based on raw-price RMSE after inverse transformation. For the local smoothers, the neighborhood search is restricted to the stable local range \(k \in \{12,\ldots,20\}\). Table~\ref{tab:results_nyc} reports the mean and standard deviation across the 5 held-out test results. Additional implementation details are collected in~\ref{app:exp_details}.
The primary comparison is between:

\begin{itemize}
    \item \textbf{KNN:} A distance-weighted \(k\)-nearest-neighbor baseline on the same tabular and geospatial predictors.
    \item \textbf{LPR / RSKLPR (Standard):} Feature-only local polynomial smoothers using the standard predictors, with RSKLPR adding the response-based robustness factor.
    \item \textbf{GC-LPR / GRC-LPR (Subway):} Context-aware local smoothers that combine a compound predictor kernel over the standard predictors and the subway network context, with GRC-LPR additionally including the response-based robustness factor:
    $$
    K_c = K_{\text{features}} \times K_{\text{subway}},
    \qquad
    \K_{\D}^* = K_c \times K_r
    $$
    \begin{itemize}
        \item $K_{\text{features}}$: A kernel on the standard tabular features.
        \item $K_{\text{subway}}$: A dedicated context kernel. This kernel maps coordinates to the nearest subway station and computes similarity based on the shortest path distance (number of stops) between stations on the subway network graph. This captures topological proximity rather than simple aerial distance. Figure~\ref{figure:subway_kernel_similarity} illustrates the spatial decay of similarity weights on the subway graph, demonstrating how the kernel captures non-Euclidean connectivity relative to a central station.
        \item $K_r$: The response-based robustness factor from RSKLPR.
    \end{itemize}
\end{itemize}

\paragraph{\textbf{Results}}
The repeated hold-out results in Table \ref{tab:results_nyc} show that the subway-context kernel improves on the feature-only LPR family and also edges out KNN on the primary RMSE metric. GRC-LPR (Subway) achieves the best mean RMSE and mean \(R^2\), with GC-LPR (Subway) a close second, while KNN attains the best mean MAE. The feature-only local smoothers trail the context-aware variants by a substantial margin. Figure~\ref{fig:nyc_exp2_rmse_boxplot} shows the RMSE distributions across the 5 held-out splits and confirms that the subway-context models remain competitive across repeats rather than only on a single favorable split.

\paragraph{\textbf{Analysis of Hyperparameters}}
The selected hyperparameters again favor very local smoothing. Across the 5 repeated hold-out runs, the subway-context models consistently selected \(p=0\), neighborhood sizes between \(k=12\) and \(k=14\), and graph distance scales of \(1.0\) or \(1.5\). LPR likewise remained highly local, usually selecting \(k=12\) or \(k=13\), while KNN selected between \(k=3\) and \(k=17\) depending on the split. Unlike Experiment~1, the context kernel does not encourage substantially larger neighborhoods. Instead, the selected subway-context models remain in the small-\(k\) regime, which is consistent with a highly irregular price surface where the graph context sharpens locality rather than broadening it.

\begin{table}[h]
\centering
\caption{Model performance on the NYC Airbnb data under the repeated 80/20 hold-out protocol. Entries report mean \(\pm\) standard deviation across 5 held-out test splits.}
\label{tab:results_nyc}
\begin{tabular}{@{}lccc@{}}
\toprule
\textbf{Model} & \textbf{RMSE (Price)} & \textbf{MAE (Price)} & \textbf{$R^2$} \\ \midrule
\textbf{GRC-LPR (Subway)} & \textbf{\$1660.35 \(\pm\) 407.59} & \$213.53 \(\pm\) 35.05 & \textbf{0.8631 \(\pm\) 0.0554} \\
GC-LPR (Subway) & \$1693.22 \(\pm\) 417.79 & \$218.58 \(\pm\) 34.94 & 0.8574 \(\pm\) 0.0583 \\
\textbf{KNN} & \$1727.02 \(\pm\) 332.86 & \textbf{\$208.93 \(\pm\) 28.36} & 0.8537 \(\pm\) 0.0482 \\
RSKLPR (Standard) & \$1995.98 \(\pm\) 301.63 & \$240.79 \(\pm\) 30.67 & 0.8056 \(\pm\) 0.0483 \\
LPR (Standard) & \$2029.73 \(\pm\) 323.32 & \$244.34 \(\pm\) 32.62 & 0.7988 \(\pm\) 0.0536 \\ \bottomrule
\end{tabular}
\end{table}

\begin{figure}[ht]
    \centering
    \includegraphics[width=0.72\textwidth]{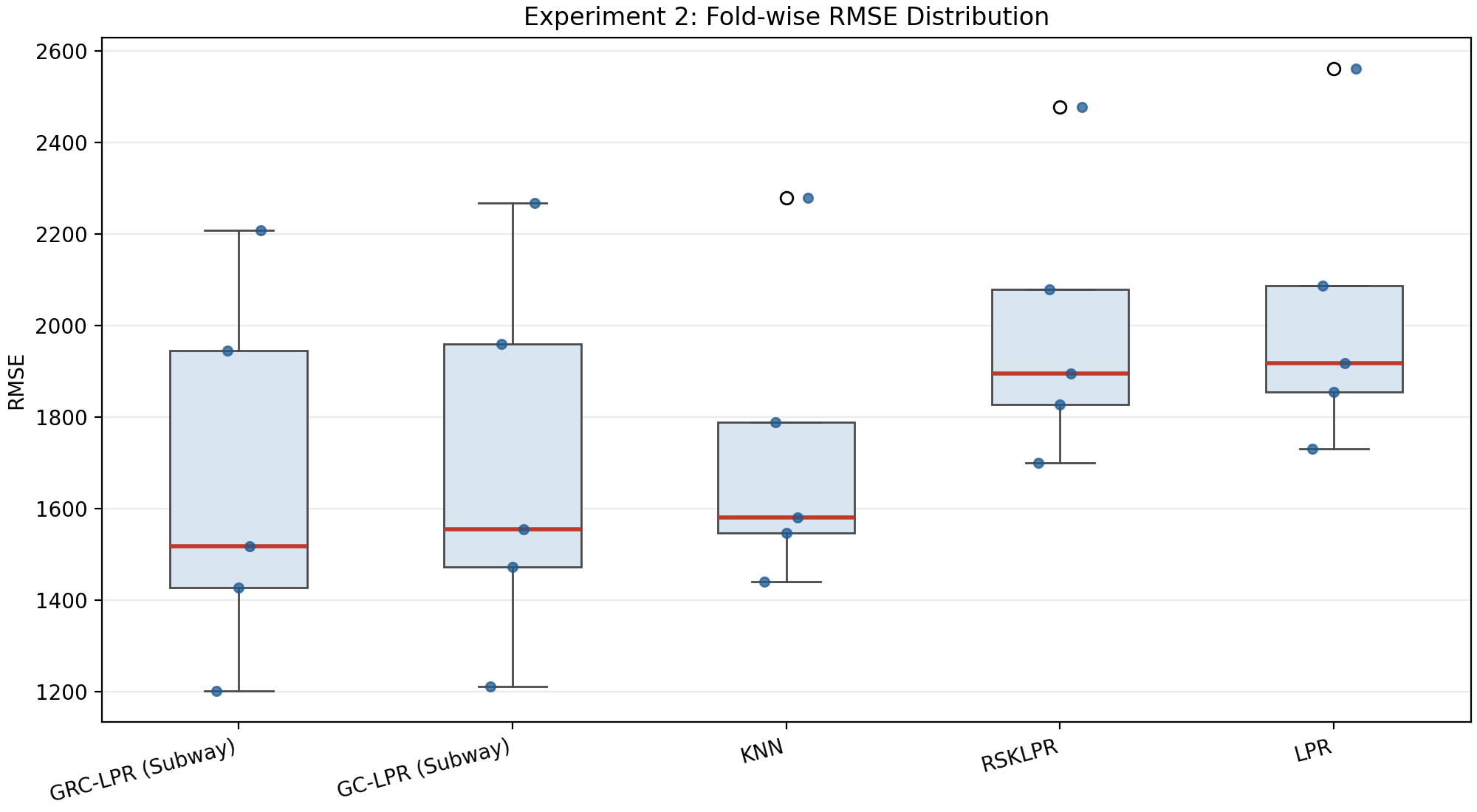}
    \caption{RMSE distributions for Experiment 2 under the repeated 80/20 hold-out protocol. The subway-context models improve on the feature-only LPR variants across repeated held-out test splits, with GRC-LPR (Subway) and GC-LPR (Subway) forming the leading pair.}
    \label{fig:nyc_exp2_rmse_boxplot}
\end{figure}

\begin{figure}[ht]
    \centering
    \begin{subfigure}[b]{0.48\textwidth}
        \centering
        \includegraphics[width=\textwidth]{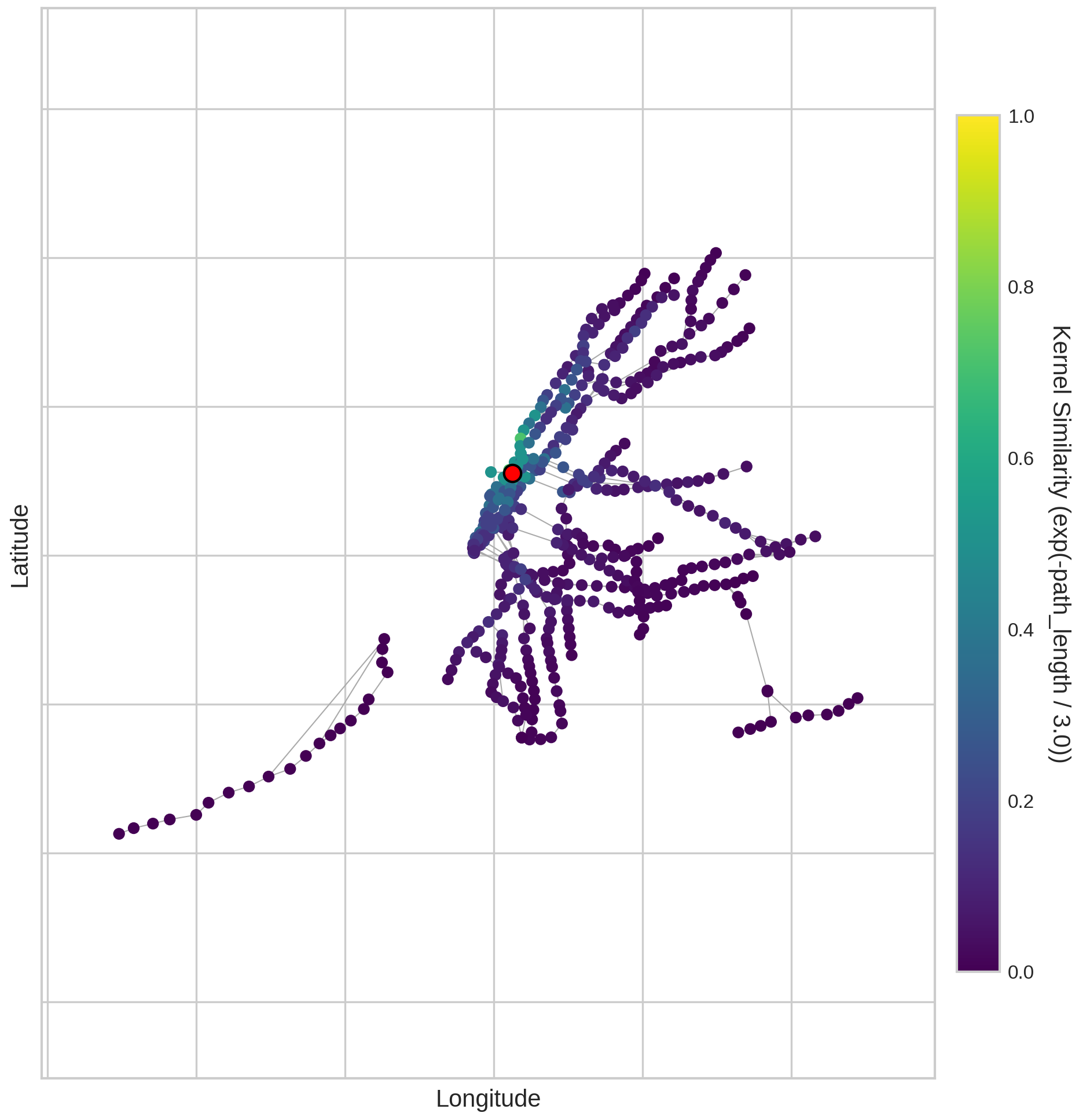}
        \caption{Full subway network view.}
        \label{fig:subway_full}
    \end{subfigure}
    \hfill
    \begin{subfigure}[b]{0.48\textwidth}
        \centering
        \includegraphics[width=\textwidth]{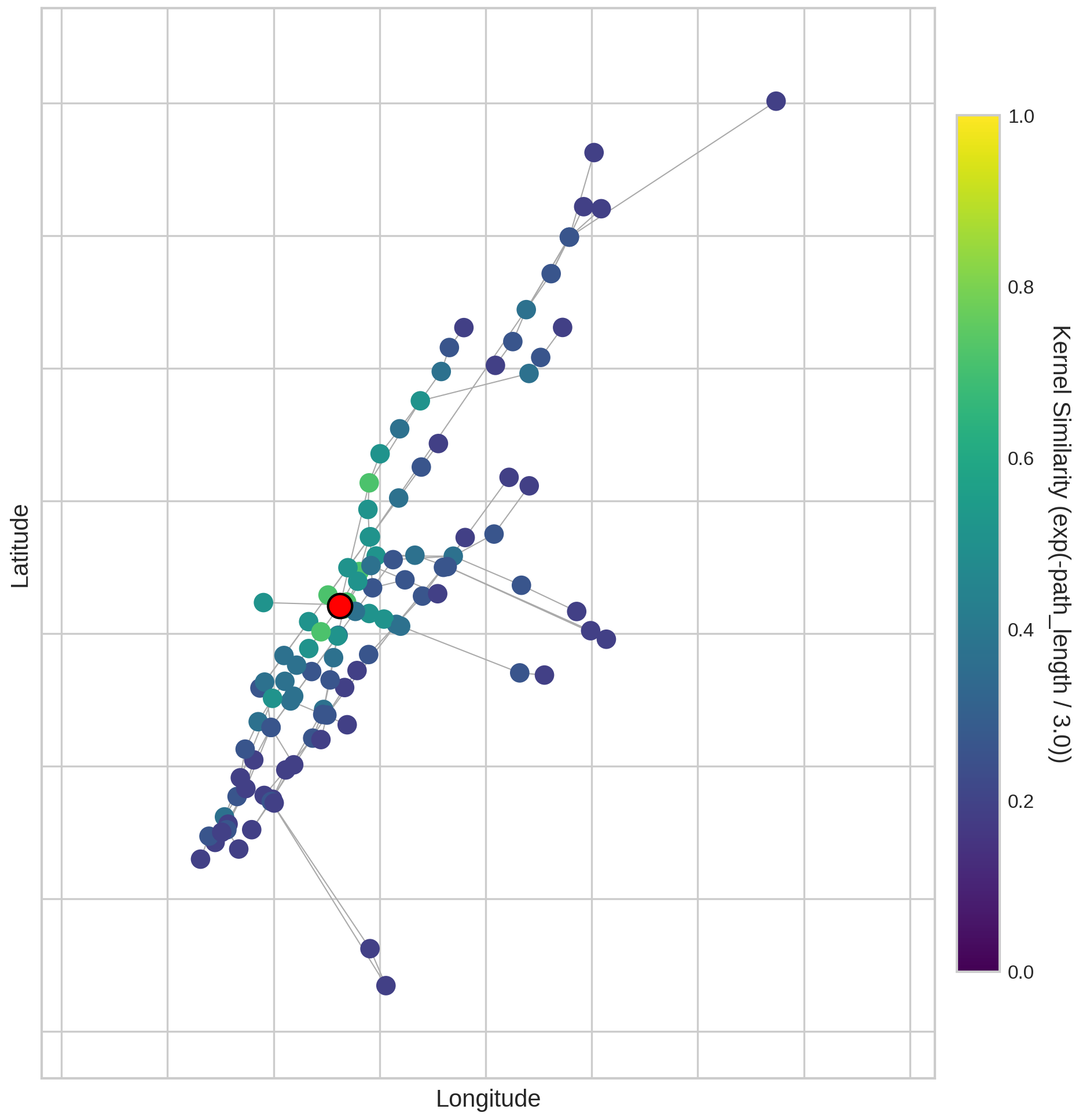}
        \caption{Zoomed-in vicinity of the reference station.}
        \label{fig:subway_zoomed}
    \end{subfigure}
    \caption{Visualizing graph-based kernel similarity for Experiment 2. The heatmaps demonstrate the similarity weights generated by $K_{\text{subway}}$ relative to the reference station ``Times Sq-42 St'' (ID 127), indicated by the large red node with a black outline. Warmer colors represent higher similarity, which decays exponentially based on the shortest path distance on the subway graph network. Panel (a) shows the distribution over the entire network, while (b) highlights the local neighborhood topology.}
    \label{figure:subway_kernel_similarity}
\end{figure}

\subsection{Experiment 3: Regression on a Node Function Defined over a Graph}
\label{s:exp_airport}

This experiment serves as the primary validation of our theoretical framework. While the previous experiments relied on real-world data where the true data generating process is unknown, here we explicitly simulate a target function $Y = m_C(Z)$ defined on the nodes of a graph. This allows us to test the specific hypothesis that \emph{structural locality} (who you are connected to) is distinct from \emph{feature similarity} (how connected you are), and that omitting the former leads to bias.

\paragraph{\textbf{Dataset and Setup}}
We construct a graph $G = (\mathcal{V}, \mathcal{E})$ representing the US domestic flight network from the public \texttt{airports.csv} airport-metadata file and \texttt{flights-airport.csv} route-count file distributed with Vega's airport-connections example \citep{vega_airports_2026,vega_flights_airport_2026}, where the route counts correspond to year 2008.
\begin{itemize}
    \item \textbf{Node Features ($Z$):} For each airport, we compute a feature vector containing standard graph centrality metrics (\texttt{PageRank}, \texttt{Betweenness}, \texttt{Degree}) and geospatial coordinates (\texttt{Latitude}, \texttt{Longitude}). Crucially, these features describe the \emph{role} of a node, but not its unique position in the topology.
    \item \textbf{Generative Process ($Y$):} We generate a synthetic "flight delay" signal that depends on both local features and network propagation. Let \(\widetilde{\mathrm{PR}}_i\), \(\widetilde{\mathrm{BW}}_i\), and \(\widetilde{\mathrm{Deg}}_i\) denote the standardized PageRank, betweenness, and degree features of node \(i\), and let \(\xi_i \stackrel{\mathrm{iid}}{\sim} N(0,1)\). The base signal is
    \begin{align*}
        \tilde{Y}^{(0)}_i &= 0.3\,\widetilde{\mathrm{PR}}_i + 0.5\,\widetilde{\mathrm{BW}}_i + 0.2\,\widetilde{\mathrm{Deg}}_i + 0.1\,\xi_i, \\
        \tilde{Y}^{(k+1)} &= 0.7\,\tilde{Y}^{(0)} + 0.3\,A\tilde{Y}^{(k)} \quad \text{for } k=0,\dots,6,
    \end{align*}
    where \(A\) is the row-normalized traffic-weighted adjacency matrix and the final target is \(Y := \tilde{Y}^{(7)}\), shifted to remain strictly positive. For clarity, the synthetic propagation uses this traffic-weighted normalized adjacency, whereas the GC/GRC-LPR context kernel itself is built on the underlying undirected airport graph via a symmetric shortest-path similarity. Further implementation details are collected in~\ref{app:exp_details}.
\end{itemize}

This process creates a scenario where two airports with similar feature vectors $Z$ (e.g., two medium-sized hubs) may have vastly different $Y$ values if they lie in different diffusion clusters. Standard LPR with a kernel only in $Z$ cannot distinguish them; GC-LPR with graph context can. Figure~\ref{fig:airport_setup} illustrates the underlying network topology and the synthetic delay function used as the prediction target.

\begin{figure}[ht]
    \centering
    \includegraphics[width=0.9\textwidth]{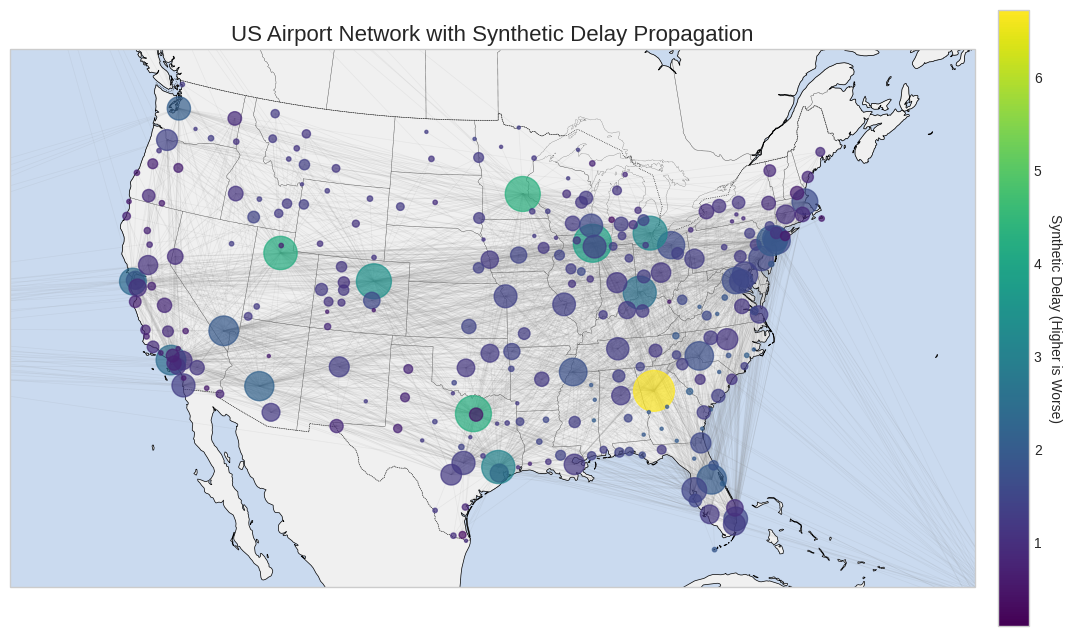}
    \caption{Visualizing the experimental setup for Experiment 3. The map displays the US airport network where edges represent flight routes and node sizes are proportional to degree. The color gradient shows the synthetic ``delay'' target function ($Y$) propagated across the network, which serves as the ground truth for the regression task.}
    \label{fig:airport_setup}
\end{figure}

\paragraph{\textbf{Models}}
We compare the same local-regression family used in the earlier experiments: KNN, LPR, RSKLPR, GC-LPR (Graph), and GRC-LPR (Graph). Experiment~3 uses 5 repeated 80/20 hold-out splits. For each repetition, hyperparameters are selected by 4-fold cross-validation on the training split, and Table~\ref{tab:results_airport} reports the mean and standard deviation of the held-out test metrics across the 5 repetitions.
\begin{itemize}
    \item \textbf{Feature-only baselines:} LPR and RSKLPR use a tricube kernel over $Z$ with the Mahalanobis metric. Across the 5 repeated hold-out runs, the selected neighborhood sizes fall in the range \textbf{79--111} with polynomial degree fixed at $p=1$. KNN remains as a generic local nonparametric baseline.
    \item \textbf{Graph-context models:} Our proposed graph-aware variants utilize the compound predictor kernel
    \[
    K_c(x, X_i) = K_{\text{feat}}(Z_i - Z_x) \times K_{\text{graph}}(C_i, C_x)
    \]
    and, when the response-based robustness factor is used,
    \[
    \K_{\D}^* = K_c \times K_r.
    \]
    Here, $K_{\text{graph}}$ is based on \emph{unweighted} shortest-path distance on the airport graph. Across the repeated hold-out runs, GC/GRC-LPR select neighborhood sizes between \textbf{111} and \textbf{135}, with graph distance scales concentrated between \textbf{0.35} and \textbf{0.40} (and one run at \(0.75\)). As shown in Figure~\ref{fig:airport_kernels}, the graph kernel assigns high weights to airports that are close in hop distance on the route network, thereby defining a non-Euclidean local neighborhood.
\end{itemize}

\begin{figure}[ht]
    \centering
    \begin{subfigure}[b]{0.48\textwidth}
        \centering
        \includegraphics[width=\textwidth]{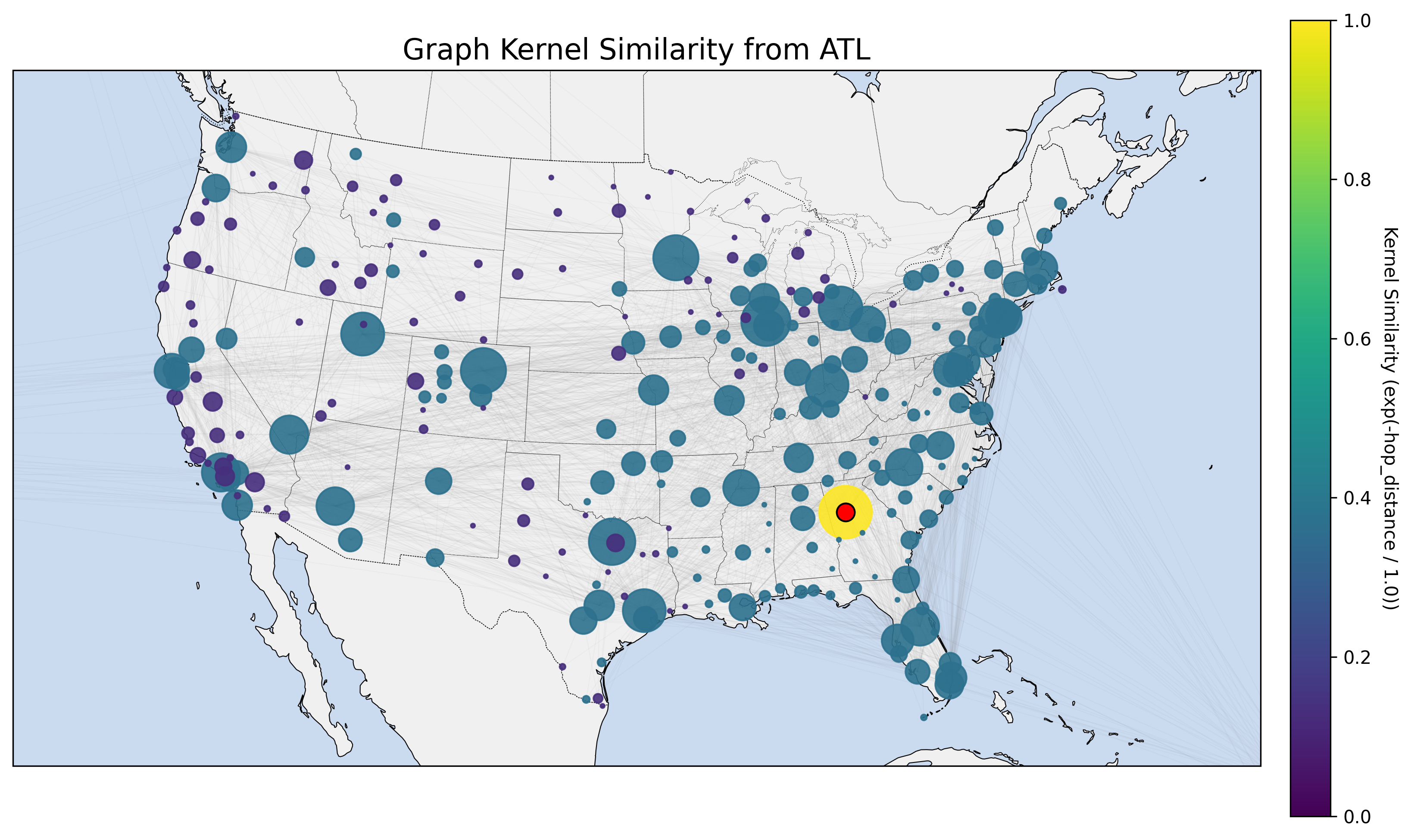}
        \caption{Similarity from Atlanta (ATL).}
        \label{fig:airport_kernel_atl}
    \end{subfigure}
    \hfill
    \begin{subfigure}[b]{0.48\textwidth}
        \centering
        \includegraphics[width=\textwidth]{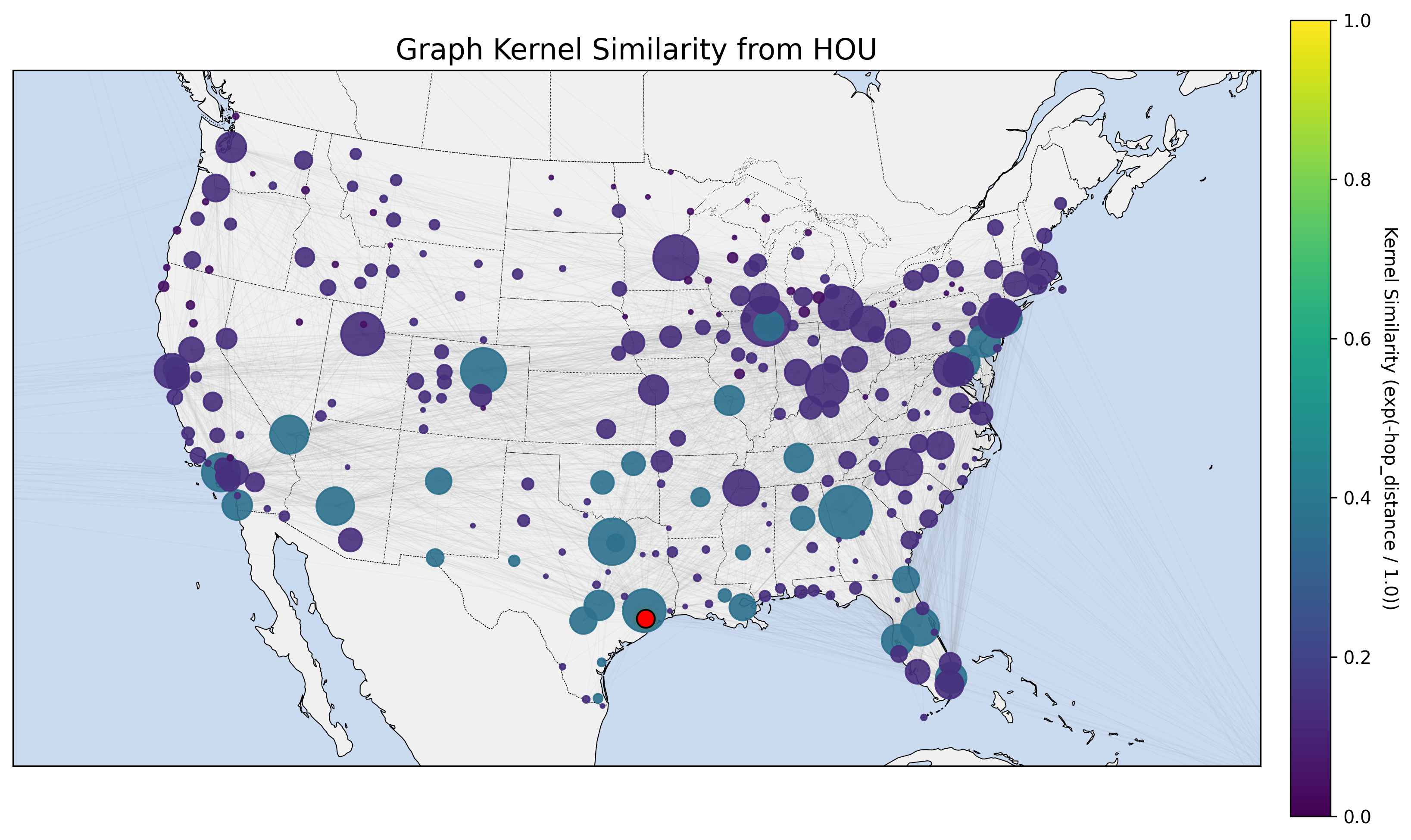}
        \caption{Similarity from Houston (HOU).}
        \label{fig:airport_kernel_hou}
    \end{subfigure}
    \caption{Visualization of the graph kernel \(K_{\text{graph}}\) used in the final Experiment~3 specification. The heatmaps show similarity based on unweighted shortest-path (hop) distance on the route network rather than Euclidean distance. The red nodes indicate the query reference airports (ATL and HOU).}
    \label{fig:airport_kernels}
\end{figure}

\paragraph{\textbf{Results}}
The repeated held-out results in Table~\ref{tab:results_airport} show that the graph-context models consistently improve on feature-only local smoothing. Because the target is generated by graph propagation on the same underlying network later used to define the context kernel, this experiment should be read as a mechanism check for graph-structured dependence rather than as a broad benchmark against all graph regression methods.

\begin{table}[h]
\centering
\caption{Experiment~3 performance under 5 repeated 80/20 hold-out splits with 4-fold cross-validation on each training split. Values are mean \(\pm\) standard deviation over held-out test metrics.}
\label{tab:results_airport}
\begin{tabular}{@{}lrrr@{}}
\toprule
\textbf{Model} & \textbf{RMSE} & \textbf{MAE} & \textbf{$R^2$} \\ \midrule
\textbf{GC-LPR (Graph)} & \textbf{0.2085 $\pm$ 0.0264} & \textbf{0.1529 $\pm$ 0.0169} & \textbf{0.8901 $\pm$ 0.0420} \\
GRC-LPR (Graph) & 0.2087 $\pm$ 0.0267 & 0.1532 $\pm$ 0.0166 & 0.8897 $\pm$ 0.0430 \\
RSKLPR & 0.2124 $\pm$ 0.0217 & 0.1590 $\pm$ 0.0158 & 0.8857 $\pm$ 0.0454 \\
LPR & 0.2127 $\pm$ 0.0218 & 0.1595 $\pm$ 0.0156 & 0.8853 $\pm$ 0.0458 \\
KNN & 0.4326 $\pm$ 0.1072 & 0.2586 $\pm$ 0.0218 & 0.5573 $\pm$ 0.1079 \\ \bottomrule
\end{tabular}
\end{table}

The performance margins are modest but consistent. GC-LPR (Graph) attains the best mean RMSE and mean \(R^2\), while GRC-LPR (Graph) is effectively tied and slightly improves mean MAE only negligibly. Relative to feature-only LPR, the graph kernel lowers mean RMSE from \(0.2127\) to \(0.2085\) and improves mean \(R^2\) from \(0.8853\) to \(0.8901\). The weak gain from the response-based robustness factor indicates that, in this synthetic setting, the main improvement comes from incorporating graph topology into the predictor kernel rather than from additional response reweighting.

\subsection{Experiment 4: Public Graph-Signal Forecasting on the Hungary Chickenpox Benchmark}
\label{s:exp_chickenpox}

To complement the synthetic airport study with an external public benchmark, we next evaluate GC-LPR on the Chickenpox Cases in Hungary benchmark of \citet{rozemberczki2021chickenpox}. This is a fixed graph-signal regression problem: the graph consists of 20 counties connected by 61 undirected edges, and the response is the weekly chickenpox count observed at each county over 521 weeks. After constructing four lagged predictors from the past weekly counts, each county-time pair becomes one regression sample, yielding 10{,}340 node-time observations with
\[
Z = (\texttt{lag\_4}, \texttt{lag\_3}, \texttt{lag\_2}, \texttt{lag\_1})
\]
and target \(Y\) equal to the next week's count.

\paragraph{\textbf{Dataset and Setup}}
This benchmark differs from Experiments~1--3 in that the graph and response are both observed rather than simulated, and the train/test protocol must preserve temporal ordering. Accordingly, Experiment~4 uses a rolling-origin design with 5 outer chronological train/test splits. On each outer split, hyperparameters are selected by 4-fold time-series cross-validation on the training window only, and the selected model is then evaluated on the next contiguous test block. Figures~\ref{fig:chickenpox_scatter}--\ref{fig:chickenpox_boxplot} summarize the held-out predictions, county-level error structure, and fold-wise RMSE distributions.

\paragraph{\textbf{Models}}
We again compare KNN, LPR, RSKLPR, GC-LPR (Graph), and GRC-LPR (Graph). The feature-only models operate only on the lagged-count vector \(Z\) using a tricube feature kernel and Manhattan distance. The graph-aware models use the compound kernel
\[
K_c(x, X_i) = K_{\text{feat}}(Z_i - Z_x) \times K_{\text{graph}}(C_i, C_x),
\]
where \(K_{\text{graph}}\) is an exponentially decaying hop-distance kernel on the county graph. The selected hyperparameters are highly stable across the 5 outer splits: LPR and RSKLPR always choose \textbf{size\_neighborhood \(=121\)}, while GC/GRC-LPR always choose \textbf{size\_neighborhood \(=361\)} and \textbf{distance\_scale \(=16.0\)}. This indicates that, on this benchmark, graph context is useful but acts as a broad structural regularizer rather than as a sharply localized graph slice.

\begin{figure}[ht]
    \centering
    \includegraphics[width=0.95\textwidth]{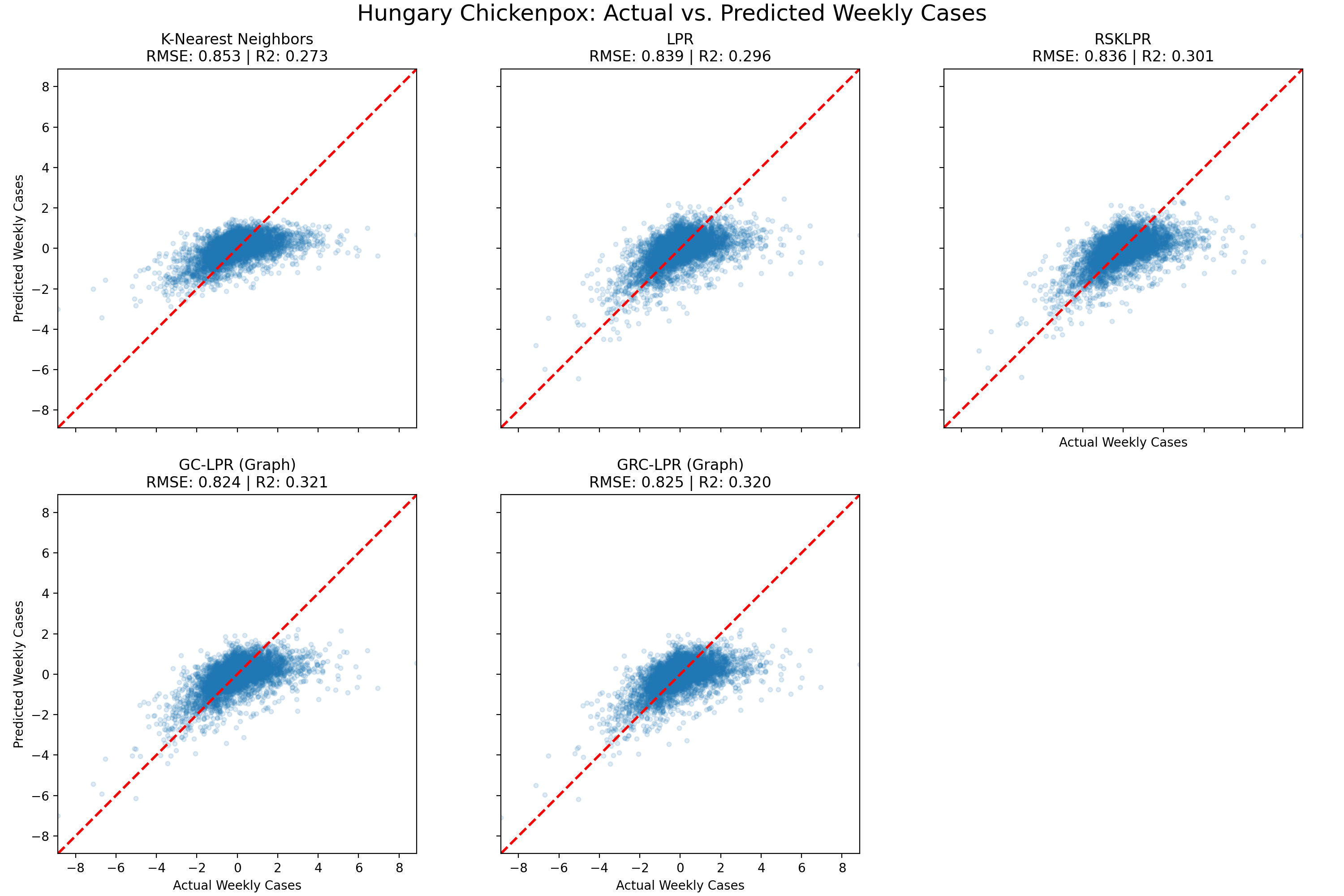}
    \caption{Experiment~4 held-out prediction scatter plots aggregated over the 5 rolling-origin test windows. Each panel compares predicted and observed county-week chickenpox counts for one model.}
    \label{fig:chickenpox_scatter}
\end{figure}

\begin{figure}[ht]
    \centering
    \includegraphics[width=0.95\textwidth]{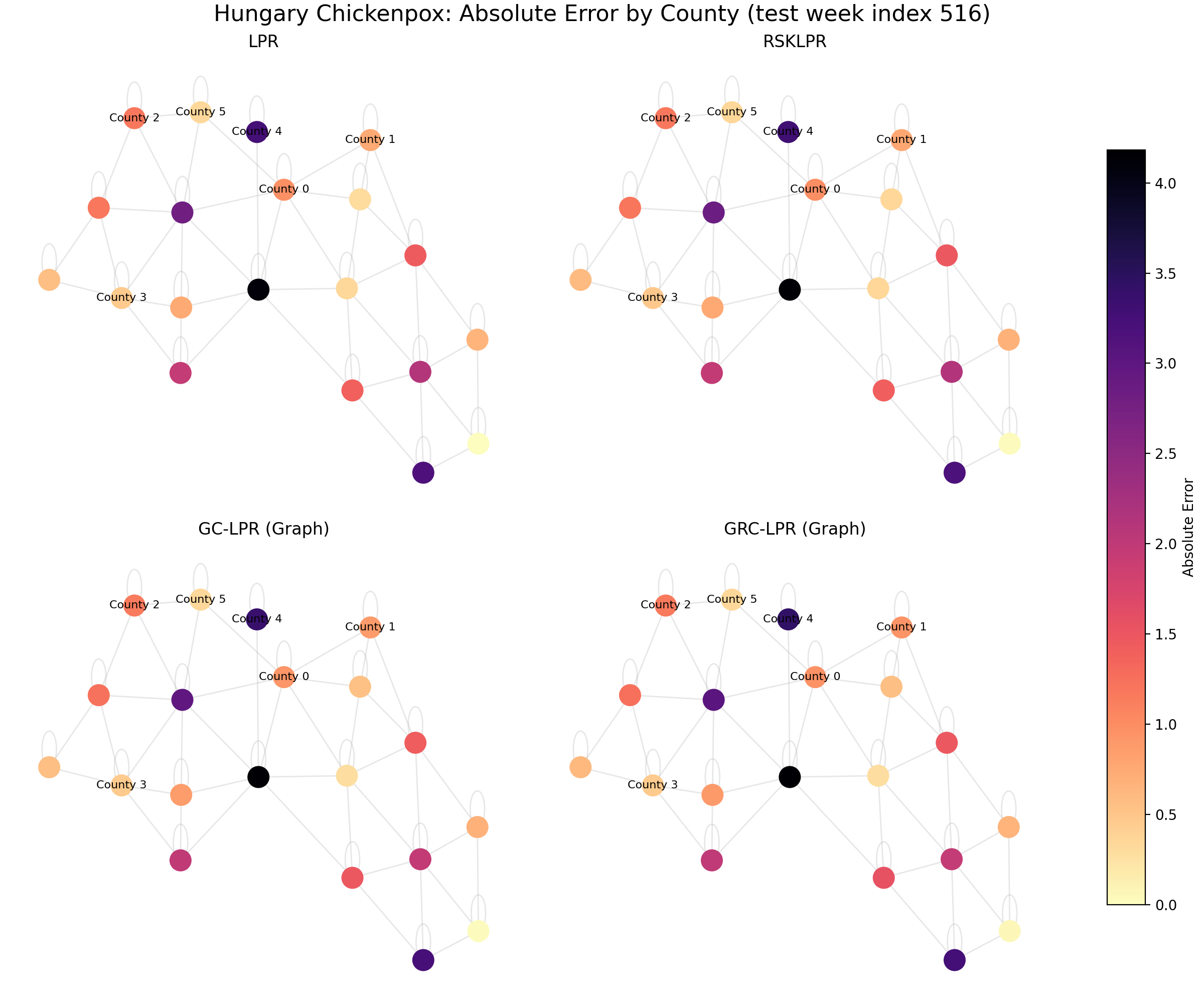}
    \caption{Experiment~4 county-level absolute-error maps on the Hungary chickenpox graph. Colors summarize held-out prediction error aggregated over the rolling-origin evaluation windows.}
    \label{fig:chickenpox_errors}
\end{figure}

\paragraph{\textbf{Results}}
Table~\ref{tab:results_chickenpox} reports the mean and standard deviation of the held-out test metrics across the 5 rolling-origin splits. The gains from graph context are modest but consistent. GC-LPR (Graph) attains the best mean RMSE and mean \(R^2\), while GRC-LPR (Graph) attains the best mean MAE by a very small margin. Relative to feature-only LPR, the hop-based graph kernel lowers mean RMSE from \(0.8347\) to \(0.8201\) and improves mean \(R^2\) from \(0.2961\) to \(0.3205\).

\begin{table}[h]
\centering
\caption{Experiment~4 performance on the public Hungary chickenpox benchmark under 5 rolling-origin splits with 4-fold time-series cross-validation on each training window. Values are mean \(\pm\) standard deviation over held-out test metrics.}
\label{tab:results_chickenpox}
\begin{tabular}{@{}lrrr@{}}
\toprule
\textbf{Model} & \textbf{RMSE} & \textbf{MAE} & \textbf{$R^2$} \\ \midrule
\textbf{GC-LPR (Graph)} & \textbf{0.8201 $\pm$ 0.0906} & 0.5370 $\pm$ 0.0569 & \textbf{0.3205 $\pm$ 0.0407} \\
GRC-LPR (Graph) & 0.8212 $\pm$ 0.0904 & \textbf{0.5367 $\pm$ 0.0572} & 0.3187 $\pm$ 0.0408 \\
RSKLPR & 0.8320 $\pm$ 0.0933 & 0.5458 $\pm$ 0.0604 & 0.3007 $\pm$ 0.0441 \\
LPR & 0.8347 $\pm$ 0.0952 & 0.5490 $\pm$ 0.0613 & 0.2961 $\pm$ 0.0473 \\
KNN & 0.8486 $\pm$ 0.0976 & 0.5557 $\pm$ 0.0608 & 0.2736 $\pm$ 0.0321 \\ \bottomrule
\end{tabular}
\end{table}

\begin{figure}[ht]
    \centering
    \includegraphics[width=0.85\textwidth]{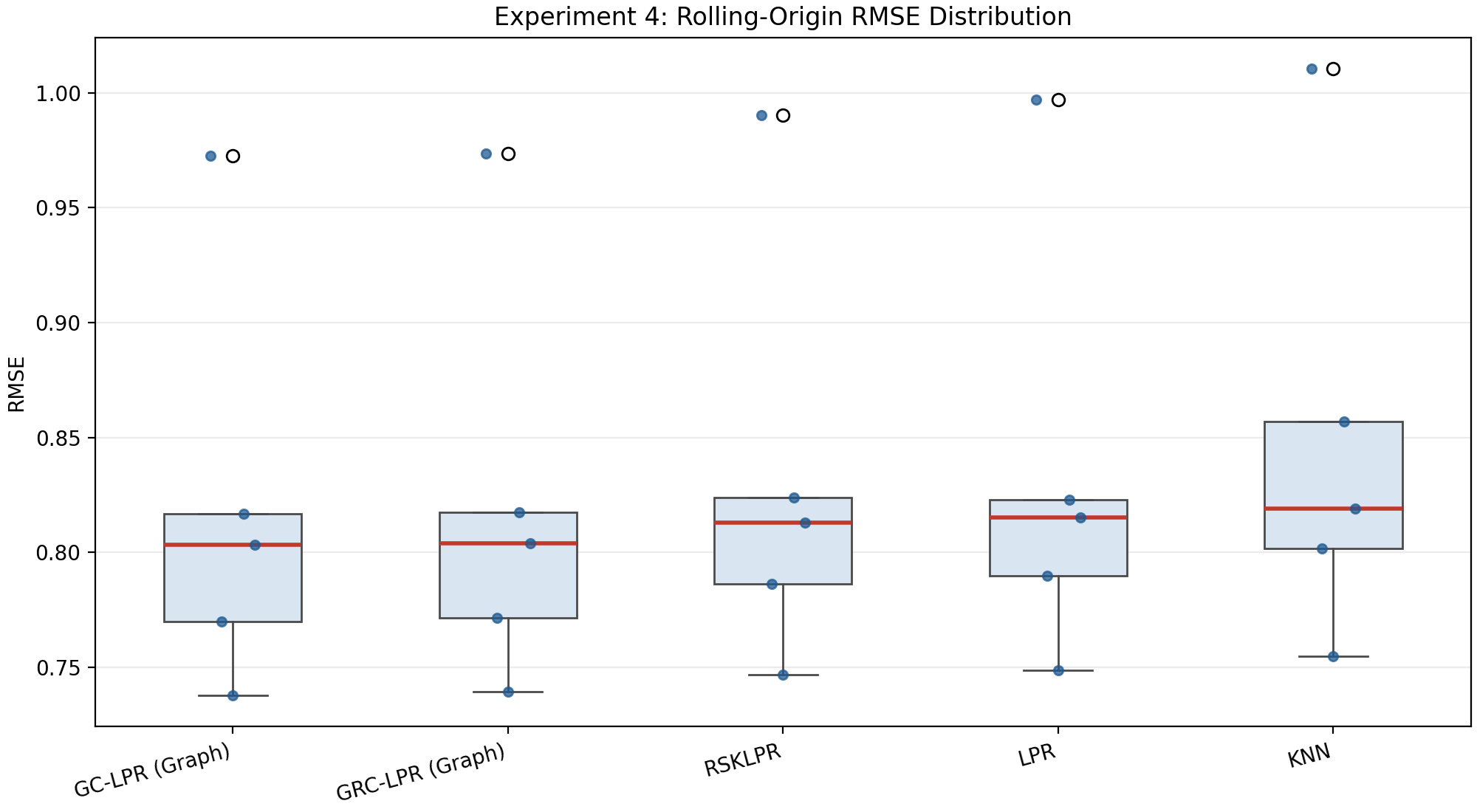}
    \caption{Fold-wise RMSE distribution for Experiment~4 under the rolling-origin protocol. The graph-aware models shift the RMSE distribution downward relative to KNN and feature-only LPR variants.}
    \label{fig:chickenpox_boxplot}
\end{figure}

Taken together, Experiments~3 and~4 show two complementary graph-context regimes. In the synthetic airport task, the graph kernel recovers an explicitly simulated topological signal. On the public chickenpox benchmark, the benefit is smaller but still systematic, which is the more relevant result for external validation. In both cases, the main improvement comes from incorporating graph structure into the predictor kernel; the additional response-based robustness factor provides only marginal gains.

\section{Conclusion}
\label{s:Conclusion}

This work presented Generalized Context-Aware Local Polynomial Regression (GC-LPR), a framework that resolves a fundamental rigidity in classical nonparametric smoothing: the requirement that the variables used to define the local neighborhood must be identical to those used in the polynomial fit. By decomposing the predictor space into a Euclidean fitting chart $\psi_0$ with coordinates $Z=\psi_0(X)$ and a context space ($C$), and bridging them via a compound product kernel, we have established a method that is both mathematically rigorous and practically flexible.

Our theoretical analysis re-frames the role of "context" not merely as a variance modifier, but as a fundamental component of the generative process $Y = m_C(Z)$. In the fixed-context-bandwidth asymptotics developed here, GC-LPR acts as a projected smoother: the context kernels (whether discrete, geospatial, or graph-based) isolate a specific slice of the data manifold, effectively "masking" irrelevant observations before the local polynomial estimation takes place in $Z$. The resulting bias-variance theory is therefore for the induced context-smoothed target $m_W$, not automatically for the point target $m_{c^\star}(z)$. This interpretation justifies the use of kernels that are discontinuous with respect to Euclidean distance, such as those defined on road networks or administrative partitions.

Empirically, we demonstrated that this decoupling yields significant performance gains. In the California Housing experiment, we showed that explicit geospatial kernels improve accuracy even when coordinates are already included as features. In the US Airport Network experiment, we demonstrated that topological contexts (flight paths) capture dependencies that standard Euclidean metrics strictly miss. On the public Hungary chickenpox benchmark, graph-aware local smoothers again improved on feature-only baselines, albeit with smaller margins, providing external validation beyond the synthetic graph experiment.

Ultimately, GC-LPR bridges the gap between interpretable classical statistics and modern geometric learning. It offers a powerful tool for domains where data is tabular but lives on an irregular domain, providing the adaptivity of graph signal processing with the well-understood local-polynomial bias-variance analysis around the induced target $m_W$.

\section*{Acknowledgements}
This research was supported by Fleet Space (\url{https://www.fleet.space}). I gratefully acknowledge my colleagues at Fleet Space for access to infrastructure and for their cooperation, which greatly assisted the research.

\clearpage

\bibliographystyle{plainnat}
\bibliography{refs}

\newpage

\appendix

\section{Proofs}
\label{app:proofs}

In this section, we prove the main population-identification and asymptotic results. The derivations follow the standard local polynomial regression arguments \citep{fan1996local} while explicitly tracking the influence of the compound context weight $W(x, u)$.

\subsection{Notation and Setup}
For simplicity of notation, throughout this appendix we define the chart coordinates used for the polynomial fit as $Z_i := \psi_0(X_i)$ and the query point as $z := \psi_0(x)$. The polynomial basis and the kernel $K_0$ operate on $Z$, while the context weights $W$ depend on the full predictors $X$.

The main text defines the estimator using the unscaled basis $r_p(Z_i-z)$. For asymptotic calculations, it is convenient to use the equivalent scaled basis $r_p(t_i)$, where
\[
t_i := H^{-1}(Z_i-z).
\]
Because composing a polynomial of total degree at most $p$ with the linear map $H$ stays in the same polynomial space, there exists an invertible matrix $A_H$ such that
\[
r_p(Z_i-z)=A_H\,r_p(t_i).
\]
Hence the main-text objective can be reparameterized by a coefficient vector $\vartheta(x)=A_H^\top\beta(x)$, and the intercept is unchanged because the constant basis term is invariant:
\[
\hat m(x)=e_1^\top\hat\beta(x)=e_1^\top\hat\vartheta(x).
\]
The scaled-basis estimator at a query point $x$ is:
\[
\hat{\vartheta}(x) = \argmin_{\vartheta} \sum_{i=1}^n \left( Y_i - r_p(t_i)^\top \vartheta \right)^2 K_0(t_i)\, W(x, X_i),
\]
where $W(x, X_i) = \prod_{j=1}^J K_j(\psi_j(x), \psi_j(X_i))$ collects all context weights. \newline

Let $R$ be the $n \times M$ design matrix with rows $r_p(t_i)^\top$, and let the weight matrix be $W_x = \text{diag}\{K_0(t_i)\, W(x, X_i)\}$. The weighted least squares solution is:
\[
\hat{\vartheta}(x) = (R^\top W_x R)^{-1} R^\top W_x Y.
\]
We define the local moment matrix $S_n$ and the projection vector $T_n$:
\[
S_n = \frac{1}{n|H|} R^\top W_x R, \quad \quad T_n = \frac{1}{n|H|} R^\top W_x Y.
\]
Conditioning on the predictors $\mathbb{X} = \{X_1, \dots, X_n\}$, the expectation is:
\[
\E[\hat{\vartheta}(x) \mid \mathbb{X}] = S_n^{-1} \E[T_n \mid \mathbb{X}] = S_n^{-1} \left( \frac{1}{n|H|} R^\top W_x \mathbf{m} \right),
\]
where $\mathbf{m} = \big(m_{C_1}(Z_1), \dots, m_{C_n}(Z_n)\big)^\top$.

\subsection{Proof of Lemma \ref{lem:mw_target} (Context-smoothed conditional target)}
\label{app:proof_mw_target}
Because Lemma \ref{lem:mw_target} is the population identification step underlying the later bias analysis, we record a slightly expanded derivation.

\begin{proof}
Fix \(x\in\mathcal{X}\) and \(u\in\mathbb{R}^d\), and write
\[
m_W:=m_W(u;x), \qquad \gamma:=\gamma_x(u).
\]
Throughout this proof, \(u\) is fixed, so \(\alpha\) and \(m_W\) are ordinary scalars in the conditional problem given \(Z=u\). Define
\[
L_u(\alpha)\coloneqq \mathbb{E}\!\left[(Y-\alpha)^2\,W(x,X)\mid Z=u\right].
\]
Under the lemma assumptions, \(0<\gamma<\infty\) and \(\mathbb{E}[Y^2W(x,X)\mid Z=u]<\infty\), so \(L_u(\alpha)\) is finite for each \(\alpha\in\mathbb{R}\).

Now decompose
\[
Y-\alpha = (Y-m_W) + (m_W-\alpha).
\]
Therefore,
\[
(Y-\alpha)^2
=
(Y-m_W)^2 + 2(Y-m_W)(m_W-\alpha) + (m_W-\alpha)^2.
\]
Multiplying by \(W(x,X)\) and taking conditional expectations given \(Z=u\), we obtain
\begin{align*}
L_u(\alpha)
&= \mathbb{E}\!\left[(Y-m_W)^2\,W(x,X)\mid Z=u\right] \\
&\quad + 2(m_W-\alpha)\,\mathbb{E}\!\left[(Y-m_W)\,W(x,X)\mid Z=u\right] \\
&\quad + (m_W-\alpha)^2\,\mathbb{E}\!\left[W(x,X)\mid Z=u\right].
\end{align*}
The cross-term vanishes because
\[
\mathbb{E}\!\left[(Y-m_W)\,W(x,X)\mid Z=u\right]
=
\mathbb{E}\!\left[Y\,W(x,X)\mid Z=u\right]
- m_W\,\mathbb{E}\!\left[W(x,X)\mid Z=u\right]
= 0,
\]
by the definition \(m_W=\mathbb{E}[Y\,W(x,X)\mid Z=u]/\gamma\) and \(\gamma=\mathbb{E}[W(x,X)\mid Z=u]\). Hence
\[
L_u(\alpha)
=
\mathbb{E}\!\left[(Y-m_W)^2\,W(x,X)\mid Z=u\right]
+ \gamma\,(\alpha-m_W)^2.
\]
Since \(\gamma>0\), the right-hand side is uniquely minimized at \(\alpha=m_W\), proving the scalar argmin statement.

For the function version, let \(a:\mathbb{R}^d\to\mathbb{R}\) be measurable. Under the conditional law given \(Z=u\), one has \(a(Z)=a(u)\) almost surely. Applying the scalar identity with \(\alpha=a(u)\) yields
\[
\mathbb{E}\!\left[(Y-a(Z))^2\,W(x,X)\mid Z=u\right]
=
\mathbb{E}\!\left[(Y-m_W(u;x))^2\,W(x,X)\mid Z=u\right]
+ \gamma_x(u)\,\big(a(u)-m_W(u;x)\big)^2.
\]
This proves Lemma~\ref{lem:mw_target}.
\end{proof}

\subsection{Proof of Proposition \ref{prop:bias} (Bias decomposition and rate)}
\label{app:proof_bias}
Throughout the proof we use the notation of Appendix~A.1.

\begin{proof}
Fix a query point $x^\star$ and write $z:=\psi_0(x^\star)$.
By Lemma~\ref{lem:mw_target}, the context-smoothed mean \(m_W(u;x^\star)\) is, for each \(u\), the unique conditional weighted least-squares target induced by the context weight \(W(x^\star,X)\). By Lemma~\ref{lem:weighted_target}, the population GC-LPR objective at \(x^\star\),
\[
\Jcal_{\mathrm{gc-lpr}}(x^\star;\beta),
\]
is, up to an additive constant independent of \(\beta\), exactly the standard local polynomial objective for the target \(m_W(\cdot;x^\star)\) under the effective design density
\[
u\mapsto f_Z(u)\,\gamma_{x^\star}(u).
\]

Under Assumptions~\ref{ass:data}--\ref{ass:positivity}, this induced problem satisfies the usual local regularity conditions near \(z\): \(m_W(\cdot;x^\star)\in C^{p+1}\) by Assumption~\ref{ass:data}, \(f_Z\) is continuous with \(f_Z(z)>0\), \(\gamma_{x^\star}(z)=\phi_1(x^\star)>0\) by Assumption~\ref{ass:positivity}, and \(\gamma_{x^\star}\) is locally Lipschitz near \(z\) by Assumption~\ref{ass:context}. Therefore the standard local polynomial bias expansion applies to the intercept estimator for this effective problem (see, e.g., \cite{fan1996local,loader1999local}), yielding
\[
\mathbb{E}\big[\hat m(x^\star)\big]-m_W(z;x^\star)
=O(\|H\|^{p+1}).
\]

Now add and subtract \(m_W(z;x^\star)\):
\[
\mathbb{E}\big[\hat m(x^\star)\big]-m_{c^\star}(z)
=
\Big(\mathbb{E}\big[\hat m(x^\star)\big]-m_W(z;x^\star)\Big)
+
\Big(m_W(z;x^\star)-m_{c^\star}(z)\Big).
\]
The first term is the polynomial approximation bias in the \(Z\)-coordinates and is \(O(\|H\|^{p+1})\); the second is the context-smoothing gap. Hence
\[
\mathbb{E}\big[\hat m(x^\star)\big]-m_{c^\star}(z)
=
O(\|H\|^{p+1})+\Big(m_W(z;x^\star)-m_{c^\star}(z)\Big).
\]

If the context kernels enforce exact selection of the context cell containing \(c^\star\), then \(m_W(z;x^\star)=m_{c^\star}(z)\) and the smoothing term vanishes. In the continuous-context case, shrinking the context bandwidths sharpens the selection and can reduce \(\big|m_W(z;x^\star)-m_{c^\star}(z)\big|\), at the cost of increased variance through a smaller effective sample size. This proves Proposition~\ref{prop:bias}.
\end{proof}

\subsection{Proof of Proposition \ref{prop:variance} (Asymptotic variance)}
\label{app:proof_variance}
\begin{proof}
Fix \(x^\star\in\mathcal{X}\) and write \(z=\psi_0(x^\star)\). By Lemma~\ref{lem:weighted_target}, the variance analysis is that of standard local polynomial regression for the target \(m_W(\cdot;x^\star)\) under the effective design density
\[
q_{x^\star}(u)\coloneqq f_Z(u)\gamma_{x^\star}(u).
\]
Define the effective residual and its weighted conditional second moment:
\[
\xi_{x^\star} \coloneqq Y-m_W(Z;x^\star),
\qquad
\nu_{x^\star}(u)\coloneqq \mathbb{E}\!\left[W(x^\star,X)^2\,\xi_{x^\star}^2\mid Z=u\right].
\]
By Lemma~\ref{lem:mw_target},
\[
\mathbb{E}\!\left[\xi_{x^\star}\,W(x^\star,X)\mid Z=u\right]=0,
\]
and Assumptions~\ref{ass:data}--\ref{ass:variance_moment} imply that, in a neighborhood of \(z\), the effective design density \(q_{x^\star}\) is continuous and bounded away from zero, while \(\nu_{x^\star}\) is locally bounded and continuous at \(z\).

Write \(\xi_{i,x^\star}\coloneqq Y_i-m_W(Z_i;x^\star)\) and decompose the score term as
\[
T_n^\xi(x^\star)\coloneqq \frac{1}{n|H|}\sum_{i=1}^n K_0(t_i)\,W(x^\star,X_i)\,r_p(t_i)\,\xi_{i,x^\star}.
\]
By the same change-of-variables and dominated-convergence arguments as in standard LPR (using the continuity of \(q_{x^\star}\) and \(\nu_{x^\star}\), together with the moment conditions in Assumption~\ref{ass:k0}), the local moment matrix from Appendix~A.1 satisfies
\[
S_n \xrightarrow{P} q_{x^\star}(z)\,M_0,
\]
and independence of \(\{(X_i,Y_i)\}_{i=1}^n\) gives
\[
\Var\!\big(T_n^\xi(x^\star)\big)
=
\frac{1}{n|H|}\,f_Z(z)\,\nu_{x^\star}(z)\,\Omega_0
\;+\;o\!\big((n|H|)^{-1}\big).
\]
Since the centered stochastic part of \(\hat\vartheta(x^\star)\) is \(S_n^{-1}T_n^\xi(x^\star)\), Slutsky's theorem yields
\[
\Var(\hat{\vartheta}(x^\star))
=
\frac{1}{n|H|}\,
\frac{\nu_{x^\star}(z)}{f_Z(z)\gamma_{x^\star}(z)^2}\,
M_0^{-1}\Omega_0M_0^{-1}
 \;+\; o\!\big((n|H|)^{-1}\big).
\]
Taking the intercept component gives
\[
\Var\{\hat m(x^\star)\}
=
\frac{1}{n|H|}\,
\frac{\nu_{x^\star}(z)}{f_Z(z)\gamma_{x^\star}(z)^2}\,
e_1^\top M_0^{-1}\Omega_0M_0^{-1}e_1
\;+\; o\!\big((n|H|)^{-1}\big).
\]
Since \(\gamma_{x^\star}(z)=\mathbb{E}[W(x^\star,X)\mid Z=z]\), this is exactly the displayed formula with
\[
\Phi(x^\star)
=
\frac{\mathbb{E}\!\left[W(x^\star,X)^2\,\big(Y-m_W(Z;x^\star)\big)^2\mid Z=z\right]}
{\mathbb{E}[W(x^\star,X)\mid Z=z]^2}.
\]
Under the model \(Y=m_C(Z)+\varepsilon\), conditioning on \(Z=z\) gives
\[
\xi_{x^\star} = \big(m_C(z)-m_W(z;x^\star)\big)+\varepsilon,
\]
so
\[
\mathbb{E}\!\left[W(x^\star,X)^2\,\xi_{x^\star}^2\mid Z=z\right]
=
\mathbb{E}\!\left[W(x^\star,X)^2\big\{\sigma^2(z,C)+\big(m_C(z)-m_W(z;x^\star)\big)^2\big\}\mid Z=z\right],
\]
because the cross term vanishes after conditioning on \((Z,C)\) and using \(\mathbb{E}[\varepsilon\mid Z,C]=0\). This gives the stated form of \(\Phi(x^\star)\).

\end{proof}

\section{Experimental Implementation Details}
\label{app:exp_details}

\subsection{Evaluation Protocols}
Experiment 1 uses the same repeated hold-out design as Experiments~2 and~3: 5 repetitions of an \(80\%/20\%\) train/test split, with hyperparameters selected by 4-fold cross-validation on each training split and reported metrics aggregated over the held-out test sets. Experiment 2 also fits the local models on \(\log(1+\texttt{price})\), but model selection and reported metrics are based on raw-price RMSE after inverse transformation of the predictions. Experiment 4 uses a rolling-origin evaluation protocol suited to graph-temporal forecasting: the county-week samples are ordered chronologically, 5 outer train/test windows are formed without shuffling, and hyperparameters are selected by 4-fold time-series cross-validation inside each training window before a single held-out evaluation on the next contiguous time block.

\subsection{Public Data Sources for Experiments 2--4}
Experiment 2 uses the archived NYC detailed listings snapshot dated February 13, 2026 from Inside Airbnb \citep{insideairbnb2026nyc}; the subway context graph is reconstructed from the MTA New York City Transit static GTFS subway feed \citep{mta_gtfs_subway_2026}. Experiment 3 uses the public \texttt{airports.csv} airport-metadata file and \texttt{flights-airport.csv} route-count file distributed with Vega's airport tutorial and datasets repository \citep{vega_airports_2026,vega_flights_airport_2026}. Experiment 4 uses the Chickenpox Cases in Hungary benchmark introduced by \citet{rozemberczki2021chickenpox}; the implementation reads the public JSON mirror distributed with PyTorch Geometric Temporal \citep{pyg_temporal_chickenpox_2026}.

\subsection{Synthetic Target Construction in Experiment 3}
The airport graph is constructed as an undirected graph on direct routes. For the synthetic target, only the PageRank, betweenness, and degree features enter the base signal, after standardization; the latitude and longitude variables are included in the regression features but not in the base-delay formula. The propagation operator uses the traffic-weighted adjacency matrix, row-normalized so each row averages the delays of neighboring airports. The update is iterated for seven steps as displayed in Section~\ref{s:exp_airport}, and the final simulated delay is shifted so that the response remains strictly positive. The graph context kernel used by GC/GRC-LPR is separate from this propagation step: it is computed from \emph{unweighted} shortest-path (hop) distance on the same undirected graph, so the context kernel captures topological locality while the synthetic target still reflects traffic-weighted diffusion.

\subsection{Public Benchmark Construction in Experiment 4}
The Hungary chickenpox benchmark is treated as a fixed graph-signal regression problem. The public dataset contains weekly chickenpox counts for 20 counties over 521 weeks on a county graph with 61 undirected edges \citep{rozemberczki2021chickenpox}. For Experiment~4, each node-time pair is flattened into one supervised sample with four lagged weekly counts as predictors and the next week's count as the response, yielding 10{,}340 samples in total. The raw JSON payload used by the implementation is the public mirror distributed with PyTorch Geometric Temporal \citep{pyg_temporal_chickenpox_2026}. The graph-aware variants use the same tricube feature kernel as the feature-only models, multiplied by an exponentially decaying hop-distance kernel on the county graph.
\end{document}